\documentclass{acmart}

\setcopyright{rightsretained}
\acmYear{2017}
\makeatletter
\@printpermissionfalse
\makeatother

\usepackage{subcaption}
\usepackage{microtype}
\usepackage{endnotes,xspace,graphicx,fancyvrb,multirow}
\usepackage{array,underscore,relsize}
\usepackage[linesnumbered]{algorithm2e}
\usepackage{wrapfig}
\usepackage{makecell,rotating}
\usepackage{mathtools}

\usepackage{pifont}

\newcommand{\sys}{\textsc{Lttp}\xspace}

\newcommand{\cc}[1]{\mbox{\texttt{#1}}}

\fvset{fontsize=\scriptsize,xleftmargin=8pt,numbers=left,numbersep=5pt}

\input{code/fmt}

\def\Snospace~{\S{}}

\newif\ifdraft\drafttrue
\newif\ifnotes\notestrue
\ifdraft\else\notesfalse\fi

\newcommand{\BH}[1]{}

\input{glyphtounicode}
\newcolumntype{R}[1]{>{\raggedleft\let\newline\\\arraybackslash\hspace{0pt}}p{#1}}

\newcommand{\squishlist}{
\begin{itemize}[noitemsep,nolistsep]
  \setlength{\itemsep}{-0pt}
}
\newcommand{\squishend}{
  \end{itemize}
}

\usepackage{tikz}

\newcommand{\add}[2]{#1 \union \{ #2 \}}

\newcommand{\assign}{\mathbin{:=}}

\newcommand{\auflia}{\textsc{Auflia}\xspace}

\newcommand{\bigland}{\bigwedge}

\newcommand{\bigunion}{\bigcup}

\newcommand{\bools}{\mathbb{B}}

\newcommand{\domain}{\mathsf{Dom}}

\newcommand{\elts}[1]{\{ #1 \}}

\newcommand{\entails}{\models}

\newcommand{\euflia}{\textsc{EufLia}\xspace}

\newcommand{\false}{\mathsf{False}}

\newcommand{\formulas}[1]{\mathsf{Forms}[#1]}

\newcommand{\ints}{\mathbb{Z}}

\newcommand{\nats}{\mathbb{N}}

\newcommand{\pset}{\mathcal{P}}

\newcommand{\sats}{\vdash}

\newcommand{\setformer}[2]{\{ #1\ |\ #2 \}}

\newcommand{\true}{\mathsf{True}}

\newcommand{\union}{\cup}

\newcommand{\upd}[3]{#1[#2 \mapsto #3]}

\newtheorem{defn}{\bf{Definition}}

\newtheorem{ex}{\bf{Example}}

\newtheorem{lemma}{\bf{Lemma}}

\newtheorem{thm}{\bf{Theorem}}

\newcommand{\alloc}[1]{\cc{#1:=new()}}

\newcommand{\allocs}{\cc{Allocs}}

\newcommand{\apps}[2]{\mathsf{Apps}[ #1, #2 ]}

\newcommand{\appsof}[1]{\mathsf{Apps}[ #1 ]}

\newcommand{\argsof}[1]{\mathsf{Args}[ #1 ]}

\newcommand{\brtargetof}[1]{\mathsf{BrTgt}[#1]}

\newcommand{\buildinsgram}{G_{\cc{BI}}}

\newcommand{\chc}[2]{\mathsf{CHC}_{ #1, #2 }}

\newcommand{\chcs}[2]{\mathsf{CHCs}_{ #1, #2 }}

\newcommand{\clauseskels}{\mathsf{ClauseSkels}}

\newcommand{\clausesof}[1]{\mathsf{Clauses}[ #1 ]}

\newcommand{\ctrof}[1]{\mathsf{Ctr}[#1]}

\newcommand{\ctrlsucc}{\mathsf{Succ}}

\newcommand{\ctxs}{\mathsf{Ctxs}}

\newcommand{\data}{\mathsf{Data}}

\newcommand{\dataedges}{E_{\mathrm{Data}}}

\newcommand{\datafields}{\cc{DFields}}

\newcommand{\datafieldseq}[3]{\mathsf{Eq\cc{DFields}}[ #1, #2, #3 ]}

\newcommand{\dataheaps}{\cc{Heaps}_D}

\newcommand{\dataloads}{\cc{Fields}_D}

\newcommand{\datastores}{\cc{Stores}_D}

\newcommand{\datastoresrcs}[1]{\mathsf{DStoreLoc}_{#1}}

\newcommand{\datauif}{\mathcal{T}_{D, \textsc{Uif}}}

\newcommand{\datavars}{\cc{DVars}\xspace}

\newcommand{\datavarseq}[2]{\mathsf{Eq\cc{DVars}}[ #1, #2 ]}

\newcommand{\derivations}[1]{\mathsf{Ders}(#1)}

\newcommand{\duality}{\textsc{Duality}\xspace}

\newcommand{\extractdeps}{\textsc{Deps}\xspace}

\newcommand{\eufliaissat}{\textsc{IsSat}\xspace}

\newcommand{\feedback}{F}

\newcommand{\fields}{\cc{Fields}}

\newcommand{\finalloc}{\cc{L}_F}

\newcommand{\gbi}{\mathcal{G}_{\cc{BI}}}

\newcommand{\headof}[1]{\mathsf{Head}[ #1 ]}

\newcommand{\heaps}{\mathsf{Heaps}}

\newcommand{\hedgesof}[1]{\mathsf{HypEdges}(#1)}

\newcommand{\idxs}{\mathsf{Idxs}}

\newcommand{\instrof}[1]{\mathsf{Instr}[#1]}

\newcommand{\instrs}{\cc{Instrs}}

\newcommand{\iseq}[3]{1:=#2=#3}

\newcommand{\isfeasible}{\textsc{IsFeas}\xspace}

\newcommand{\isempty}{\mathsf{Infeas}}

\newcommand{\initloc}{\cc{L}_I}

\newcommand{\isnil}[2]{\cc{#1:=isNil(#2)}}

\newcommand{\instrat}[3]{\mathsf{Instr}[ #1 ](#2, #3)}

\newcommand{\invtrees}[1]{\mathcal{I}[\mathcal{G}]}

\newcommand{\lang}{\cc{Lang}\xspace}

\newcommand{\lblinstrs}{\cc{LblInstrs}}

\newcommand{\load}[3]{\cc{#1:=#2->#3}}

\newcommand{\loads}{\cc{Loads}}

\newcommand{\locs}{\cc{Locs}\xspace}

\newcommand{\locsym}{\mathsf{Loc}}

\newcommand{\modelof}[1]{m_{#1}}

\newcommand{\modelsof}[1]{\mathsf{Models}[ #1 ]}

\newcommand{\negfrags}[1]{\rightsquigarrow_{#1}^-}

\newcommand{\nil}{\mathsf{nil}}

\newcommand{\niltests}{\cc{NilTests}}

\newcommand{\nodesof}[1]{\mathsf{Nodes}[ #1 ]}

\newcommand{\nuAll}{\nu_{\mathrm{All}}}

\newcommand{\objctxs}{\mathsf{Ctxs}_O}

\newcommand{\objeqs}{\cc{ObjEqs}}

\newcommand{\objfields}{\cc{OFields}}

\newcommand{\objfieldseq}[3]{\mathsf{Eq\cc{OFields}}[ #1, #2, #3 ]}

\newcommand{\objframeof}[1]{#1_{O}}

\newcommand{\objheaps}{\mathsf{Heaps}_O}

\newcommand{\objs}{\mathsf{Objs}}

\newcommand{\objloads}{\cc{Lds}_O}

\newcommand{\objstores}{\cc{Stores}_O}

\newcommand{\objvars}{\cc{OVars}\xspace}

\newcommand{\objvarseq}[2]{\mathsf{Eq\cc{OVars}}[ #1, #2 ]}

\newcommand{\ord}[1]{<_{#1}}

\newcommand{\params}{\mathsf{Params}}

\newcommand{\pathof}{\mathsf{DerPath}}

\newcommand{\pathsof}[1]{\mathsf{Paths}[#1]}

\newcommand{\posfrags}[1]{\rightsquigarrow_{#1}^+}

\newcommand{\postfrag}{\mathsf{Post}}

\newcommand{\prefrag}{\mathsf{Pre}}

\newcommand{\prelocof}[1]{\mathsf{PreLoc}[ #1 ]}

\newcommand{\queryof}[1]{\mathsf{Query}[ ]}

\newcommand{\relof}[1]{\mathsf{Rel}[ #1 ]}

\newcommand{\runbg}{\mathcal{T}_{\cc{Lang}}}

\newcommand{\runchcs}{\mathsf{CHCs}_{\cc{Lang}}}

\newcommand{\runsof}[1]{\mathsf{Runs}[#1]}

\newcommand{\sgn}[1]{\mathit{sgn}_{#1}}

\newcommand{\skeletons}{\mathsf{Skels}}

\newcommand{\skelpreds}{\mathsf{R}_{ \mathsf{ Skel} }}

\newcommand{\sketch}{\textsc{Sketch}\xspace}

\newcommand{\solvechc}{\textsc{SolveCHC}\xspace}

\newcommand{\states}{\mathsf{States}}

\newcommand{\store}[3]{\cc{#1->#2:=#3}}

\newcommand{\stored}[1]{\mathsf{Stored}[ #1 ]}

\newcommand{\stores}{\cc{Stores}}

\newcommand{\succsym}{\mathsf{Succ}}

\newcommand{\sympath}{\mathsf{SymPath}}

\newcommand{\symrel}[1]{\mathsf{SymRel}[ #1 ]}

\newcommand{\symvaltrans}[1]{\mathsf{SymRel}_V[ #1 ]}

\newcommand{\synchc}{\textsc{SynGrammar}\xspace}

\newcommand{\synskeleton}{\textsc{SynSkeleton}\xspace}

\newcommand{\sysaux}{\textsc{Lttp}'\xspace}

\newcommand{\tnow}{\cc{t}_{\mathrm{now}}}

\newcommand{\transrelof}[1]{\rightarrow_{#1}}

\newcommand{\valctxs}{\mathsf{Ctxs}_V}

\newcommand{\valinstrs}{\cc{Instrs}_V}

\newcommand{\valtransrelof}[1]{\rightarrow_{#1}^V}

\newcommand{\values}{\mathsf{Vals}}

\newcommand{\valvars}{\cc{Vars}_V}

\newcommand{\vars}{\cc{Vars}\xspace}

\newcommand{\varsof}[1]{\mathsf{Vars}[ #1 ]}

\newcommand{\vocab}{V}

\newcommand{\cmark}{\ding{51}}%
\gdef\therev{e9f5ce0}
\gdef\thedate{2017-10-07 12:26:57 -0400}

\begin{document}

\title{Proofs as Relational Invariants of Synthesized Execution Grammars}

\ifdefined\DRAFT
 \pagestyle{fancyplain}
 \lhead{Rev.~\therev}
 \rhead{\thedate}
 \cfoot{\thepage\ of \pageref{LastPage}}
\fi

\author{Caleb Voss}
\affiliation{\institution{Georgia Institute of Technology}}
\author{David Heath}
\affiliation{\institution{Georgia Institute of Technology}}
\author{William Harris}
\affiliation{\institution{Georgia Institute of Technology}}

\date{}

\begin{abstract}
  The automatic verification of programs that maintain unbounded
  low-level data structures is a critical and open problem.
  Analyzers and verifiers developed in previous work can synthesize
  invariants that only describe data structures of heavily restricted
  forms, or require an analyst to provide predicates over program data
  and structure that are used in a synthesized proof of correctness.

  In this work, we introduce a novel automatic safety verifier of
  programs that maintain low-level data structures, named \sys.
  \sys synthesizes proofs of program safety represented as a grammar
  of a given program's control paths, annotated with invariants that
  relate program state at distinct points within its path of
  execution.
  \sys synthesizes such proofs completely automatically, using a novel
  inductive-synthesis algorithm.

  We have implemented \sys as a verifier for JVM bytecode and applied
  it to verify the safety of a collection of verification benchmarks.
  Our results demonstrate that \sys can be applied to automatically
  verify the safety of programs that are beyond the scope of
  previously-developed verifiers.
\end{abstract}

\maketitle
\thispagestyle{fancy} 

\section{Introduction}
\label{sec:introduction}
Automatically verifying that a given program satisfies a desired
safety property is a fundamental problems of program verification.
Recent work has seen the development of powerful program verifiers
that operate
automatically~\cite{bjorner13,gupta11,heizmann10,henzinger02,henzinger04,mcmillan06}.
Such verifiers can often determine if practical programs satisfy
properties concerning their control flow and facts over a bounded
collection of data values~\cite{ball11}.

However, verifying the safety of programs that maintain unbounded
low-level data structures remains an open problem.
A significant body of previous work has developed \emph{shape
  analyzers}~\cite{holik13,dudka16} that, given a program \cc{P},
synthesize invariants of the reachable heaps of \cc{P} represented in a
particular shape domain, such as three-valued logical
structures~\cite{loginov05,reps04,sagiv02} or separation-logic
formulas~\cite{calcagano11,distefano06,reynolds02,yang08};
the invariants synthesized by such analyzers can potentially imply
facts about program states that establish that \cc{P} is safe.
Another body of work has developed automatic program
verifiers~\cite{albarghouthi15,balaban05,flanagan02,rakamaric07,lahiri06,bouillaguet07,rondon10,dams03,drews16,itzhaky13,itzhaky-ba14,itzhaky-bj14}
and decision procedures~\cite{perez11,seshia03} that directly attempt
to determine if \cc{P} is safe by attempting to synthesize sufficient
invariants in such domains.

Unfortunately, all such approaches suffer from at least one of several
critical limitations.
In particular, they either are only able to represent invariants over
heaps of restricted
forms~\cite{balaban05,dams03,lahiri06,drews16,itzhaky13,itzhaky-ba14,itzhaky-bj14} or
require an analyst to manually provide
abstractions~\cite{reps04,sagiv02,calcagano11,distefano06,reynolds02,yang08,holik13,dudka16,perez11,seshia03,albarghouthi15,flanagan02,rakamaric07,rondon10,dams03,itzhaky-bj14}
or candidate inductive invariants~\cite{bouillaguet07,itzhaky-ba14}.

The main contribution of this work is an automatic verifier, named
\sys, that attempts to determine the safety of a program that may
maintain low-level data structures.
\sys satisfies two key features that distinguish it from previous
approaches.
First, it can potentially prove the safety of programs by establishing
inductive invariants that relate multiple low-level structures
maintained by a program, each of which need not necessarily have a
pre-specified shape, such as list or tree.
Second, \sys performs synthesis of these invariants completely
automatically, without requiring an analyst to provide predicates over
values or data-structure shapes from which to attempt to synthesize a
proof.

We designed \sys to satisfy both of the above features by developing
two key technical insights.
The first insight is that the safety of a program that maintains
low-level \emph{data structures} can often be established by
invariants that relate only values bound to variables (i.e.,
\emph{local data}) at \emph{multiple control points}.
Each proof structure synthesized by \sys is an annotated \emph{path
  grammar} of the program---i.e., a graph grammar in which the yield
of each derivation is a program path.
Each non-terminal $A$ of the grammar is annotated with a formula that
relates only the local values of a tuple of program points generated
by each application of $A$.
Such a formula is in the combination of a theory that axiomatizes the
program's data operations with the theory of uninterpreted functions.

The second insight is that proofs as relational invariants of graph
grammars can potentially be synthesized automatically, using an
\emph{inductive
  synthesizer}~\cite{alur13,solar-lezama05,solar-lezama06,solar-lezama07,solar-lezama08}.
The inductive synthesizer implemented by \sys, given a program \cc{P},
iteratively maintains a set of enumerated control paths of \cc{P}.
In each iteration, \sys synthesizes a candidate graph grammar
$\mathcal{G}$ from the structure of \cc{P} and enumerated paths by
reduction to constraint solving.
\sys attempts to synthesize relational invariants of $\mathcal{G}$ by
reduction to logic programming~\cite{bjorner13,rummer13}.
If a logic-programming solver determines that no such invariants
exist, then \sys collects from the solver an unexplored control path
which it uses to refine the proposed grammar in subsequent
iterations. 

We have developed an implementation of \sys that verifies programs
represented in Java Virtual Machine (JVM) bytecode.
We evaluated \sys by using it to attempt to verify the safety of a
collection of challenging problems for shape analyses and verifiers,
some of which are adapted from those presented in the SV-COMP verifier
competition~\cite{svcomp17}.

The results demonstrate that \sys is powerful enough to express
non-local invariants, combining information separated in time by
unboundedly many execution steps and separated in heap space by
unboundedly many field accesses.
Without manual assistance, \sys can prove unary and binary properties
over lists, reachability properties, and even correlative properties
between two disjoint data structures.

In summary, the design of \sys builds on and contributes to multiple
topics in program analysis and verification.
First, \sys constitutes a shape verifier with strengths that are
distinct from all shape verifiers presented in previous work, though
it does not necessarily subsume such verifiers, as discussed in
\autoref{sec:features}.
Second, \sys demonstrates that techniques from \emph{relational
  verification}~\cite{benton04,barthe11,barthe13} which previously
have been applied to prove properties of multiple programs by relating
states \emph{across} multiple runs, can be applied to prove shape
properties by relating states \emph{within} program runs.
Third, \sys demonstrates that techniques from inductive synthesis,
which have previously been applied to synthesize correct programs as
completions of incomplete programs, can also be used to synthesize
alternate representations of complete programs that are amenable to
verification.

The rest of this paper is organized as follows.
\autoref{sec:overview} illustrates the operation of \sys by example.
\autoref{sec:background} reviews previous work on which \sys is based,
\autoref{sec:pf-structures} describes the proofs synthesized by \sys,
and %
\autoref{sec:verifier-algo} describes \sys's inductive-synthesis
algorithm.
\autoref{sec:evaluation} describes our empirical evaluation of \sys.
\autoref{sec:related-work} compares \sys to related work, and %
\autoref{sec:conclusion} concludes.


\section{Overview}
\label{sec:overview}
This section illustrates \sys by example.
\autoref{sec:ex-program} introduces an example program as a
verification problem.
\autoref{sec:ex-pf} introduces the proof of \cc{buildInspect}
synthesized by \sys.
\autoref{sec:ex-syn} describes how \sys synthesizes the proof
automatically.

\subsection{\cc{buildInspect}: maintaining a low-level queue}
\label{sec:ex-program}
\begin{figure}[t]
  \centering
  \begin{minipage}{.45\textwidth}
    \centering
    \input{code/buildInspect.java}
  \end{minipage}
  \qquad
  \begin{minipage}{.45\textwidth}
    \centering
    \includegraphics[width=\linewidth]{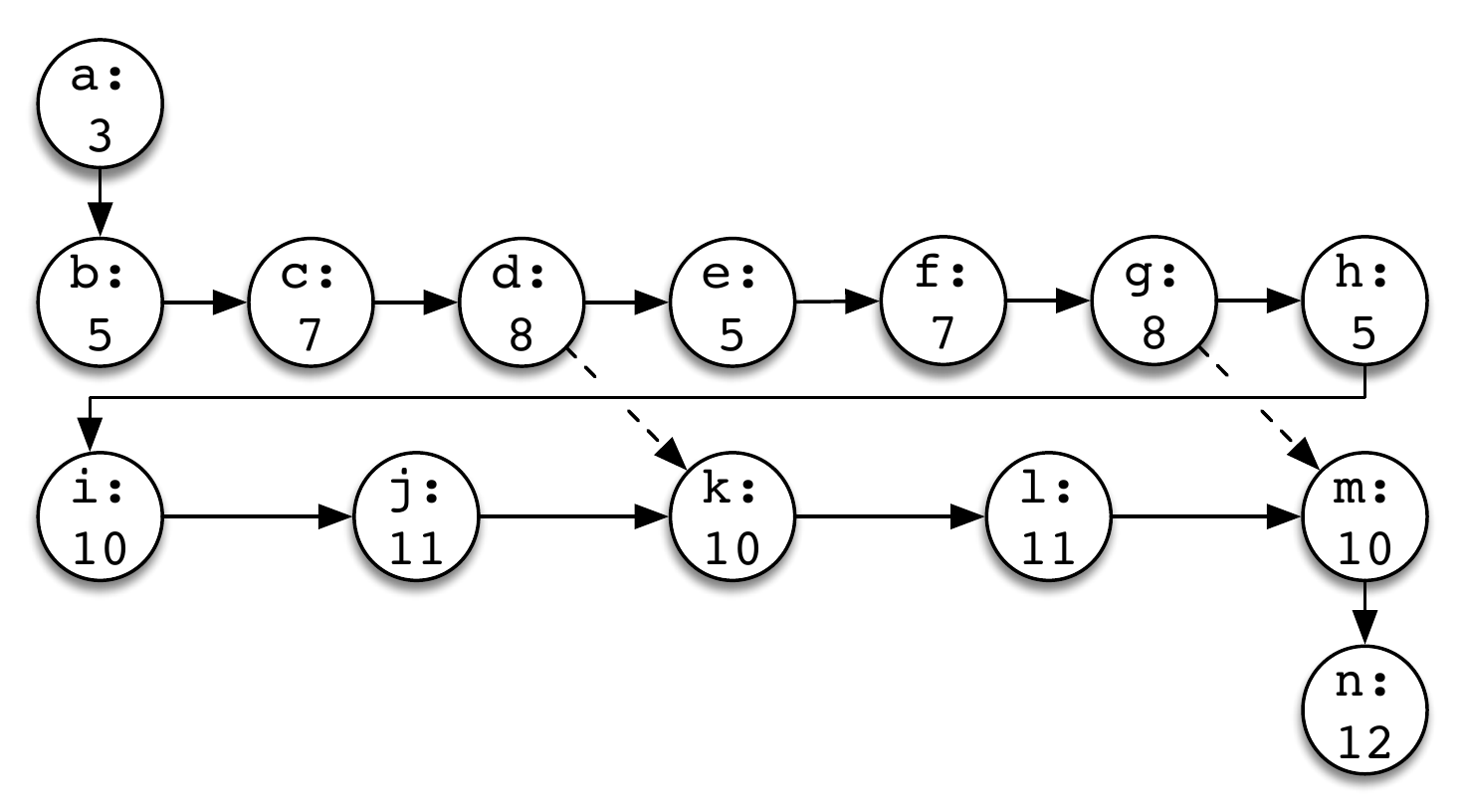}
  \end{minipage}
  \begin{minipage}[t]{.45\textwidth}
    \centering
    \caption{\cc{buildInspect}: constructs a queue with \cc{num}
    elements and traverses it from head to tail.} %
    \label{fig:running-ex-code}
  \end{minipage}
  \qquad
  \begin{minipage}[t]{.45\textwidth}
    \centering
    \caption{A control path of \cc{buildInspect} containing two
      iterations of the loops, including data-dependence
    edges.} %
    \label{fig:ex-path}
  \end{minipage}
\end{figure}

\autoref{fig:running-ex-code} contains a program, named
\cc{buildInspect}, that maintains a low-level queue.
The queue is represented as a linked list of \cc{Element} objects
(line \cc{1}).
\cc{buildInspect} first constructs a single \cc{Element} object and
binds it to variables that store both the queue's head (line \cc{3})
and tail (line \cc{4}).
\cc{buildInspect} then iteratively creates new \cc{Element}s and adds
them to the end of the queue (lines \cc{5}---\cc{8}).
\cc{buildInspect} then iterates over the elements in the queue,
storing each element in \cc{elt} (lines \cc{10}---\cc{11}).
Finally, \cc{buildInspect} asserts that, on completion of the loop in
lines \cc{10}---\cc{11}, \cc{elt} stores the tail element of the queue
(line \cc{12}).

Each execution of \cc{buildInspect} satisfies the assertion at line
\cc{12}.
The key invariant for the loop over lines \cc{5}---\cc{8} establishes
that the \cc{Element} object stored in \cc{elt} reaches the object
stored in \cc{tail} over some number $i$ of dereferences of the
\cc{next} pointer.
The key invariant for the loop over lines \cc{10}---\cc{11}
establishes that the \cc{Element} object stored in \cc{elt} reaches
the \cc{Element} object stored in \cc{tail} over $num-i$ dereferences
of the \cc{next} pointer, where $n$ is the length of the queue. 
Unfortunately, it is challenging to design a verifier that
can express such invariants and infer them automatically from a given
program without supplied candidate invariants.

\subsection{A proof of \cc{buildInspect}'s safety as relational
  invariants of a path grammar}
\label{sec:ex-pf}
A key insight behind our approach is that proofs of the safety of a
program \cc{P} can be represented as a graph grammar that generates
control paths of \cc{P}, annotated with invariants that relate states
at different path points when they co-occur in the same grammar rule.
In this section, we give a grammar $\gbi$ of \cc{buildInspect}'s
control paths (\autoref{sec:ex-grammar}) and relational invariants of
$\gbi$ that represent a proof of the safety of \cc{buildInspect}
(\autoref{sec:ex-invs}).

\subsubsection{A grammar of \cc{buildInspect}'s control paths}
\label{sec:ex-grammar}
\autoref{fig:ex-path} contains the control path of \cc{buildInspect}
that includes two iterations of the loop at lines \cc{5}---\cc{8} and
two iterations of the loop at lines \cc{10}---\cc{11}.
In \autoref{fig:ex-path}, each node is an instance of a control
location, annotated with its line number. The line number of each node
indicates the program has reached, but not executed, the instruction
on that line.
Control steps are depicted as solid edges.
Each point that stores an object is connected to the point that loads
the stored object by a dashed line.

\begin{figure*}[t]
  \centering
  \includegraphics[width=.95\linewidth]{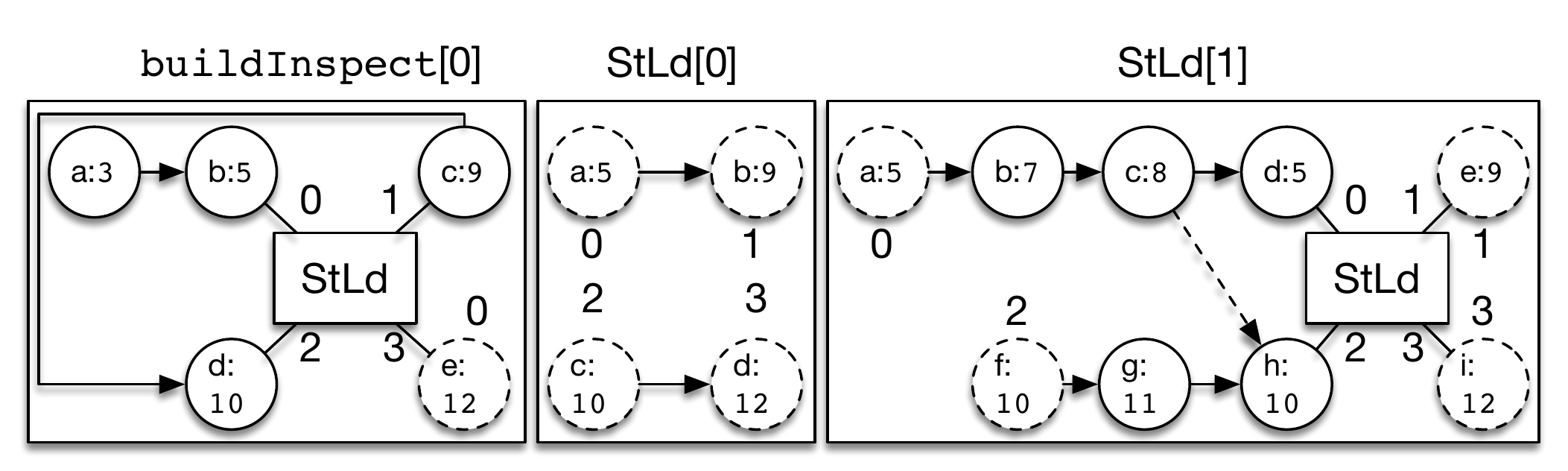}
  \caption{A graph grammar that generates \cc{buildInspect}'s control
    paths and has relational invariants that prove
    the safety of \cc{buildInspect}.
    The grammar contains three clauses, each labeled with their head
    relation (either \cc{buildInspect} or $\mathsf{StLd}$).}
  \label{fig:build-inspect-grammar}
\end{figure*}
A graph grammar $\gbi$ that generates the control paths of
\cc{buildInspect} is depicted in \autoref{fig:build-inspect-grammar}.
$\gbi$ contains two relations, \cc{buildInspect} and
$\mathsf{StLd}$, with \cc{buildInspect} the starting relation.
$\gbi$ consists of three clauses, one for \cc{buildInspect} and two
for \cc{StLd}.
Each clause is depicted as a hypergraph. The nodes represent program
control points and are labeled with an alphabetic character for
reference and with the corresponding line number in the program.
The edges consist of control-dependence edges, depicted as solid
edges;
the dashed edges depict data dependencies, and are described in
\autoref{sec:ex-invs}.
Each (ordered) hyperedge $h$ is represented as a box labeled with a
relation, and with the nodes in $h$ labeled with their
index in $h$.
The hyperedge that each clause defines is represented implicitly, with
the nodes in the hyperedge dashed and labeled with their index in the
hyperedge.

The starting relation of $\gbi$, \cc{buildInspect}, has a single
clause, $\cc{buildInspect}[0]$, that generates fragments of a control
path that model the initial step of \cc{buildInspect} and the step of
\cc{buildInspect} from the exit of the loop at lines \cc{5}---\cc{9}
to the entry of the loop at lines \cc{10}---\cc{11}.

The relation \cc{StLd} simultaneously derives a pair of paths,
corresponding to the two loops of the program. 
$\gbi$ contains two clauses for $\mathsf{StLd}$.
Clause $\mathsf{StLd}[0]$ generates a pair of loop exit steps.
Clause $\mathsf{StLd}[1]$ generates both a step through the loop at
lines \cc{5}---\cc{9} that allocates a new \cc{Element} object and
stores it at the \cc{next} field of a bound object, and a step through
the loop at lines \cc{10}---\cc{11} that loads an \cc{Element} object.
Clause $\mathsf{StLd}[1]$ recursively includes an instance of the
$\mathsf{StLd}$ relation, which can be further expanded in order to
generate the rest of each loop.

The control path in \autoref{fig:ex-path} is generated by
$\gbi$.
In particular, it is generated by applying the following sequences of
rules: $\cc{buildInspect}[0]$, $\cc{StLd}[1]$, $\cc{StLd}[1]$,
$\cc{StLd}[0]$.

It should be noted that the grammar shown in
\autoref{fig:build-inspect-grammar} is a simplification of the grammar
required to prove the safety of \cc{buildInspect}.
In particular, the depicted grammar only allows expansions of
\cc{StLd} such that the two loops iterate the same number of times.
It is true that every actual execution has this property, but this
fact is not known \emph{a priori}.
Therefore, the grammar must allow paths
with different numbers of iterations in each loop so that the
underlying model checker can discover such an invariant through
counterexamples.

\subsubsection{Relational invariants of \cc{buildInspect}'s path
  grammar}
\label{sec:ex-invs}
A second key observation that motivates our approach is that a proof
of the safety of a program \cc{P} can be represented as some grammar
$\mathcal{G}$ of the paths of \cc{P} paired with relational invariants
over the location instances in the interface of each of the relations
of $\mathcal{G}$.
In particular, a proof of the safety of \cc{buildInspect} can be
represented as the path grammar $\gbi$ (given in
\autoref{sec:ex-grammar}), paired with relational invariants over
instances of control locations in the interface of \cc{buildInspect}
and $\mathsf{StLd}$ that establish that at the $0$th interface node of
\cc{buildInspect}, $\cc{elt} = \cc{tail}$.

The relational invariant for $\mathsf{StLd}$ establishes that if
\cc{tail} at index $0$ is equal to \cc{elt} at index $2$ and if the
point in the path at which the \cc{next} field of \cc{tail} at index
$1$ was stored is the point at which the \cc{next} field of \cc{elt}
at index $2$ was stored, then \cc{tail} at index $1$ is \cc{elt} at
index $3$. It also establishes that \cc{i} at index $0$ is equivalent
to \cc{i} at index $2$, indicating that the loops will iterate the
same number of times.
The relational invariant for $\mathsf{StLd}$ can be expressed as the
following formula, where integer subscripts indicate interface node
indices and $\cc{next}$ indicates a function which maps each object to
the path point at which it was last updated:
\begin{align}
  (\cc{tail}_0 = \cc{elt}_2 \land
  \cc{next}_1(\cc{tail}_0) = \cc{next}_2(\cc{elt}_2) \implies
  \cc{tail}_1 = \cc{elt}_3)
  \land \cc{i}_0 = \cc{i}_2
  \label{eqn:stld}
\end{align}

Formula~\ref{eqn:stld} entails that each run of \cc{buildInspect}
satisfies its assertion at line \cc{12}.
Furthermore, it can be proved to hold over all tuples of path points
in all paths generated by $\gbi$, by induction on the
derivations of $\gbi$.
In particular,
\textbf{(1)} in clause $\mathsf{StLd}[0]$, the semantic constraints of
the instructions on generated control steps entail
Formula~\ref{eqn:stld} with the variables at index $0$, $1$, $2$, and
$3$ replaced with variables that represent state at points $a$, $b$,
$c$, and $d$ respectively.
\textbf{(2)} In clause $\mathsf{StLd}[1]$, semantic constraints of the
instructions on generated control steps, combined with
Formula~\ref{eqn:stld} with the variables at index $0$, $1$, $2$, and
$3$ replaced with variables that model state at points $d$, $e$, $h$,
and $i$, entail Formula~\ref{eqn:stld} with variables at index $0$,
$1$, $2$, and $3$ replaced with variables that represent state at path
points $a$, $e$, $f$, and $i$.

\subsection{Proving safety of \cc{buildInspect} automatically}
\label{sec:ex-syn}
\sys, given a program \cc{P} with an error location \cc{L}, attempts
to prove that \cc{L} is unreachabale in \cc{P} by synthesizing a path
grammar that has relational invariants that prove that \cc{L} is
unreachable in \cc{P}.
To prove that \cc{buildInspect} always satisfies the assertion at line
\cc{12}, \sys synthesizes the path grammar of \cc{buildInspect} given
in \autoref{sec:ex-grammar} and the relational
invariants of the grammar given in \autoref{sec:ex-invs}.

\sys, given program \cc{P} and error location \cc{L}, attempts to
synthesize a proof that \cc{L} is unreachable in \cc{P} as relational
invariants by performing an inductive synthesis
algorithm~\cite{alur13, srivastava11, gulwani11, itzhaky10}.
The algorithm maintains an initially empty set of enumerated control
paths of \cc{P}.

\begin{figure*}[t]
  \centering
  \includegraphics[width=.9\linewidth]{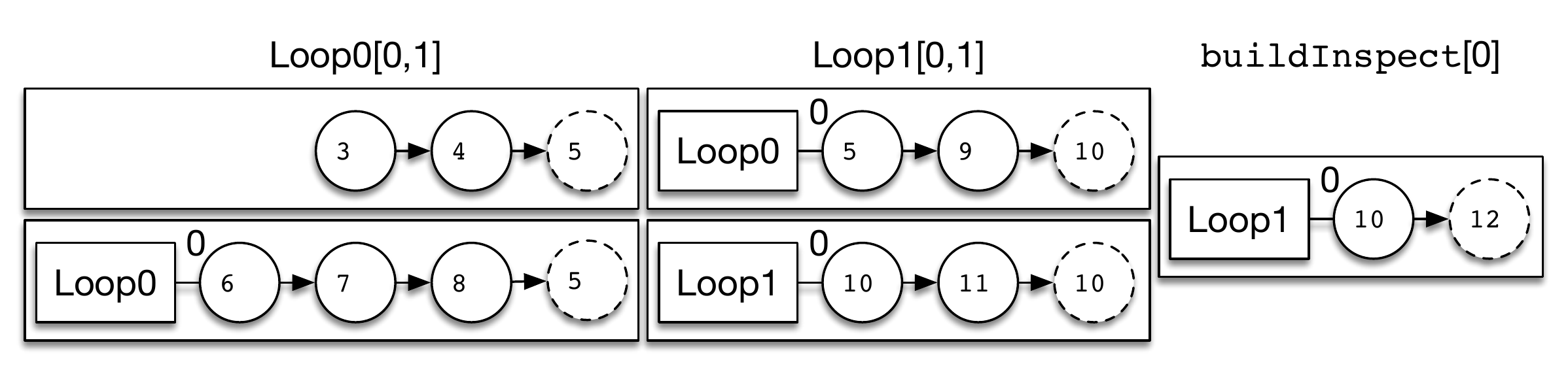}
  \caption{An alternative graph grammar $\buildinsgram'$ that
    generates cc{buildInspect}'s control paths, but does not have
    relational invariants that prove the safety of \cc{buildInspect}.}
  \label{fig:build-inspect-cfg}
\end{figure*}
In each iteration, \sys synthesizes a graph grammar $\mathcal{G}$ that
derives all control paths of \cc{P} and that admits relational
invariants proving the safety of all paths enumerated.
As the set of enumerated paths is initially empty, the first grpah
grammar synthesized by \sys is likely a simple grammar whose
structure directly corresponds to the control-flow graph of
\cc{buildInspect}, such as the grammar $\buildinsgram'$ given in
\autoref{fig:build-inspect-grammar}, in which each path is derived
left-recursively. 
E.g., it could synthesize the graph grammar $\buildinsgram'$ given in
\autoref{fig:build-inspect-grammar}, in which each relation
derives a single path left-recursively.
As a heuristic, \sys could initially synthesize path grammars that use
right-recursive definitions, the program's syntactic structure, or
other sub-structures of the program's control-flow graph, such as
components of its Bourdoncle decomposition~\cite{bourdoncle93}.
However, no such fixed grammar can serve as a proof of correctness in
general, and must be iteratively restructured.

After \sys synthesizes a graph grammar $\mathcal{G}$, it attempts to
determine if $\mathcal{G}$ has relational invariants that represent a
proof of the safety of \cc{P}.
\sys reduces this problem to solving a system $\mathcal{S}$ of
\emph{Constrained Horn Clauses} (see \autoref{sec:chcs}), and runs a
solver on the generated system as a black box.
If the solver determines that $\mathcal{S}$ has a solution, then the
solution contains relational invariants of $\mathcal{G}$ that prove
safety of \cc{P}.

Otherwise, if the solver provides a counter-derivation $D$ of
$\mathcal{S}$ that has no solution, then $D$ defines a
control path that cannot be proved safe by relational invariants of
$\mathcal{G}$.
For example, when given $\buildinsgram'$, \sys determines that some
derivation of $\buildinsgram'$, such as the one corresponding to the
control path $p$ given in \autoref{fig:ex-path}, has no solution.
Such a derivation has no solution because nodes of the path that are
the sites of matching stores and loads do not co-occur in the clauses
of $\buildinsgram'$.

\sys then inspects the control path of the unsolvable derivation and
determines if it is truly unsafe.
If so, then \sys determines that \cc{P} is unsafe.
Otherwise, as in the case of $p$, \sys determines that the path is
safe, adds it to the set of enumerated paths, and recurses.
In \sys's next iteration, it will synthesize a path grammar distinct
from $\buildinsgram'$ that admits relation invariants proving the
safety of $p$.
\sys eventually synthesizes $\buildinsgram$, which proves safety for
all paths.

\section{Background}
\label{sec:background}
In this section, we review foundations on which \sys is built.
In \autoref{sec:lang}, we define the low-level imperative language
targeted by \sys.
In \autoref{sec:logic-chcs}, we review concepts from formal logic and
logic programming used in the design of \sys.

\subsection{A target imperative language}
\label{sec:lang}
In this section, we define the structure (\autoref{sec:structure}) and
semantics (\autoref{sec:semantics}) of \sys's target language.

\subsubsection{Structure}
\label{sec:structure}
A program is a set of instructions, each labeled with a control
location and a target control location.
The space of all \emph{control locations} is denoted \locs, with
distinguished \emph{initial} location and \emph{final} location
$\initloc, \finalloc \in \locs$.
The finite spaces of value variables and object variables are denoted
$\datavars$ and $\objvars$, respectively;
their union is denoted $\vars = \datavars \union \objvars$.
The finite spaces of \emph{data fields} and \emph{object fields} are
denoted $\datafields$ and $\objfields$, respectively.
All spaces of mutually disjoint.
\sys can be applied to programs in languages with infinite spaces of
variables and fields, because it only must consider the finite sets of
variables and fields that occur in a given program.
We assume that the spaces of variables and fields themselves are
finite in order to simplify the presentation of \sys.

The space of \emph{data} instructions is denoted $\valinstrs$, and
includes standard operations of Boolean and linear integer arithmetic.
Let $\cc{p}, \cc{q} \in \objvars$, $\cc{x} \in \datavars$, $\cc{f} \in
\datafields$, and $\cc{g} \in \objfields$ be arbitrary elements.
\load{y}{p}{f} is a \emph{data load} and \store{p}{f}{y} is a
\emph{data store};
\load{q}{p}{g} is an \emph{object load} and \store{p}{g}{q} is an
\emph{object store}.
The spaces of all data loads, data stores, object loads, and object
stores are denoted $\dataloads$, $\datastores$, $\objloads$, and
$\objstores$, respectively.
\alloc{p} is an \emph{allocation}, %
\isnil{x}{p} is a \emph{$\nil$ test}, and %
\iseq{x}{p}{q} is an \emph{object-equality test};
the space of all allocations, $\nil$ tests, and object-equality tests
are denoted \allocs, \niltests, and \objeqs, respectively.
The space of all instructions is denoted
\[ \instrs = \valinstrs \union \dataloads \union \datastores \union
\objloads \union \objstores \union \allocs \union \niltests \union
\objeqs \]

A pre-location, instruction, and branch-target location is a
\emph{labeled instruction};
i.e., the space of labeled instructions is denoted $\lblinstrs = \locs
\times \instrs \times \locs$.
For each labeled instruction $\cc{i} \in \lblinstrs$, the
pre-location, instruction, and post-location of $i$ are denoted
$\prelocof{\cc{i}}$, $\instrof{\cc{i}}$, and $\brtargetof{\cc{i}}$,
respectively.

A program is a set of labeled instructions.
I.e., the space of programs is $\lang = \pset(\lblinstrs)$, where for
any set $S$, $\pset(S)$ denotes the powerset of $S$.
For each $\cc{P} \in \lang$ and all $\cc{L}, \cc{L}' \in \locs$, there
may be at most one $\cc{i} \in \lblinstrs$ such that
$\prelocof{\cc{i}} = \cc{L}$ and $\brtargetof{ \cc{i} } = \cc{L}'$.
In such a case, \cc{i} is denoted alternatively as
$\instrat{ \cc{P} }{ \cc{L} }{ \cc{L}' }$.

\subsubsection{Semantics}
\label{sec:semantics}
A run of a program \cc{P} is a sequence of states generated by a
sequence of labeled instructions in which adjacent instructions have
matching target locations and pre-locations.
The space of \emph{data} is $\data = \bools \union \ints$, %
the space of \emph{non-nullable objects} is a countably-infinite space
$\objs'$, and the space of \emph{objects} is $\objs =
\add{\objs'}{\nil}$.

The space of \emph{local data states} is $\valctxs = \datavars \to
\data$, the space of \emph{local object states} is $\objctxs =
\objvars \to \objs$, and the space of \emph{local states} is $\ctxs =
\valctxs \times \objctxs$.
The space of \emph{data heaps} is $\dataheaps = \objs' \times
\datafields \to \data$, the space of \emph{object heaps} is $\objheaps
= \objs' \times \objfields \to \objs$, and the space of \emph{heaps}
is $\heaps = \dataheaps \times \objheaps$.
The space of states is $\states = \ctxs \times \heaps$.

\begin{table}[t]
  \begin{tabular}{| r @{} l | c | r  @{} r  @{} l @{} l | }
    \hline
    \multicolumn{2}{|c|}{ \multirow{2}{*}{ \textbf{Instruction} } } & 
    \textbf{Updated} &
    \multicolumn{4}{c|}{ \multirow{2}{*}{ \textbf{Updated Value} } } \\
    & & \textbf{Components} & & & & \\
    \hline
    \hline
    \multicolumn{2}{|c|}{$\cc{i} \in \valinstrs$} & $\sigma_D$ &
    \multicolumn{4}{c|}{$\sigma_D' \in \valctxs, \sigma_D \valtransrelof{
        \cc{i} } \sigma_D'$} \\
    \hline
    \cc{x:=} & \cc{p->f} & $\sigma_D$ &
    $\sigma_D[$ & 
    $\cc{x} \mapsto$ & 
    $H(\sigma_O(\cc{p}), \cc{f})$ & $]$ \\
    \hline
    \cc{p->f:=} & \cc{x} & $H_D$ &
    $H_D[$ & 
    $(\sigma_O(\cc{p}), \cc{f}) \mapsto$ &
    $\sigma_D(\cc{x})$ & $]$ \\
    \hline
    \cc{q:=} & \cc{p->f} & $\sigma_O$ & 
    $\sigma_O[$ &
    $\cc{q} \mapsto$ &
    $H_D( \sigma_D(\cc{p}), \cc{f})$ & $]$ \\
    \hline
    \cc{p->f:=} & \cc{q} & $H_O$ & 
    $H_O[$ &
    $(\sigma_O(\cc{p}), \cc{f}) \mapsto$ &
    $\sigma_O(\cc{q})$ & $]$ \\
    \hline
    \cc{x:=} & \cc{p=q} & $\sigma_D$ &
    $\sigma_D[$ & $\cc{x} \mapsto$ &
    $\sigma_O(\cc{p}) = \sigma_O(\cc{q})$ & $]$ \\
    \hline
    \cc{x:=} & \cc{isNil(p)} & $\sigma_D$ & 
    $\sigma_D[$ & $\cc{x} \mapsto$ & 
    $\sigma_O(\cc{p}) = \nil$ & $]$ \\
    \hline
    \multirow{3}{*}{\cc{p:=}} & 
    \multirow{3}{*}{\cc{new()}} & $\sigma_O$ &
    $\sigma_O[$ & $\cc{p} \mapsto$ & $o \notin \domain(H_O) =
    \domain(H_D)$ & $]$ \\
    & & $H_D$ &
    $H_D[$ & $\elts{ o } \times \datafields \mapsto$ & $0$ & $]$ \\
    & & $H_O$ &
    $H_O[$ & $\elts{ o } \times \objfields \mapsto$ & $\nil$ & $]$ \\
    \hline
\end{tabular}
\caption{Resulting post-state of instructions from, for 
  $\sigma_D \in \valctxs$, $\sigma_O \in \objctxs$, 
  $H_D \in \dataheaps$, and $H_O \in \objheaps$, the pre-state
  $((\sigma_D, \sigma_O), (H_D, H_O))$.
  In each instruction, $\cc{p}, \cc{q} \in \objvars$, $\cc{x} \in
  \datavars$, $\cc{d} \in \datafields$, and $\cc{f} \in \objfields$.
  In the entry for \alloc{p}, for each map $m$, $\domain(m)$ denotes
  the \emph{domain} of $m$.}
\label{tab:instr-semantics}
\end{table}
For each $\cc{i} \in \valinstrs$, there is a transition relation
$\valtransrelof{\cc{i}} \subseteq \valctxs \times \valctxs$.
The transition relation of a value instruction need not be total:
thus, labeled instructions can implement control branches using
instructions that act as \cc{assume} instructions.
The transition relation of instructions is defined based on their
structure.
The transition relation of each instruction, defined over an arbitrary
pre-state, is given in \autoref{tab:instr-semantics}.

For each labeled instruction $\cc{i} \in \lblinstrs$, the transition
relation of \cc{i} is the transition relation of the instruction of
\cc{i};
i.e., $\transrelof{ \cc{i} } = \transrelof{ \instrof{ \cc{i} } }$.
A program state $\sigma$ may not be the source of any entry in the
transition relation if it binds an object variable \cc{p} to $\nil$
and executes an instruction that attempts to load from or store to
$\nil$.
For the remainder of this paper, we do not consider a stuck state to
be an error state;
programs can be transformed so that stuck states are error states in
our formulation.

For each finite space $N$, %
$E \subseteq N \times N$ such that $(N, E)$ is a sequential graph and
each %
$\lambda : N \to \locs$, $(N, E, \lambda)$ is a \emph{control path}.
For each $\cc{P} \in \lang$, if $(N, E, \lambda)$ is such that for all
path edges $(n, n') \in E$, there exists an instruction
$\instrat{\cc{P}}{\lambda(n)}{\lambda(n')}$, then $(N, E, \lambda)$ is
a control path of \cc{P}.
The space of control paths of \cc{P} is denoted $\pathsof{\cc{P}}$.
For each $\cc{p} \in \pathsof{ \cc{P} }$, the nodes of \cc{p} are
denoted $\nodesof{ \cc{p} }$.

A run of a program $\cc{P}$ is a control path $p$ of $\cc{P}$ and a
map from each node of $p$ to a state such that states associated with
adjacent nodes satisfy the transition relation of a corresponding
instruction of \cc{P}.
\begin{defn}
  \label{defn:runs}
  For $\cc{P} \in \lang$, %
  finite space $N$, $E \subseteq N \times N$, and $\lambda: N \to
  \locs$ such that $(N, E, \lambda) \in \pathsof{ \cc{P} }$, %
  let $\sigma: N \to \states$ be such that %
  for all $n, n' \in N$ with $(n, n') \in E$, %
  $\sigma(n) \transrelof{ \instrat{ \cc{P} }{ \lambda(n) }{
      \lambda(n') } } \sigma(n')$.
  Then for $p = (N, E, \lambda)$, $(p, \sigma)$ is a \emph{run} of $p$
  in \cc{P}.
\end{defn}
For $p \in \pathsof{\cc{P}}$, the runs of $p$ in \cc{P} are denoted
$\runsof{p}$.
For $r \in \runsof{p}$, $p$ is denoted alternatively as $\pathof{r}$.
If $\runsof{p}$ is empty, then $p$ is \emph{infeasible}.
The runs of \cc{P} are the runs of all paths of \cc{P};
i.e., the runs of \cc{P} are denoted $\runsof{\cc{P}} = \bigunion_{p
  \in \pathsof{\cc{P}}} \runsof{p}$.
If $\runsof{\cc{P}}$ is empty, then \cc{P} is \emph{safe}.
The core problem addressed in this work is, given program $\cc{P}$, to
determine if $\cc{P}$ is safe.

\subsection{Constraint solving and Constrained Horn Clauses}
\label{sec:logic-chcs}
\subsubsection{Formal logic}
\label{sec:logic}
The quantifier-free fragment of the theory of the combinations of
linear arithmetic and uninterpreted functions is denoted $\euflia$.
For each space of logical variables $X$, the space of $\euflia$
formulas over $X$ is denoted $\formulas{X}$.
For each formula $\varphi \in \formulas{X}$, the set of variables that
occur in $\varphi$ (i.e., the \emph{vocabulary} of $\varphi$) is
denoted $\vocab(\varphi)$.
For formulas $\varphi_0, \ldots, \varphi_n, \varphi \in \formulas{X}$,
the fact that $\varphi_0, \ldots, \varphi_n$ \emph{entail} $\varphi$
is denoted $\varphi_0, \ldots, \varphi_n \entails \varphi$.
The models of variables $X$ are denoted $\modelsof{X}$.
The fact that a model $m$ \emph{satisfies} a formula $\varphi$ is
denoted $m \sats \varphi$.
\sys uses a decision procedure for \euflia, named \eufliaissat.
We assume that for each $\cc{i} \in \valinstrs$, \sys may access some
formula $\symvaltrans{ \cc{i} } \in$ $\formulas{ \euflia }$.

\subsubsection{Constrained Horn Clauses}
\label{sec:chcs}
Constrained Horn Clauses are a class of logic-programming problems
that formulate problems in program
verification~\cite{bjorner13,flanagan03,gupta11,rummer13}.
\paragraph{Structure}
A Constrained Horn Clause is a body, consisting of uninterpreted
relational predicates applied to logical variables and a constraint,
and a head application.
Relational predicates are symbols associated with arities.
\begin{defn}
  \label{defn:rel-preds}
  For each space of symbols $\mathcal{R}$ and function $a: \mathcal{R}
  \to \nats$, $(\mathcal{R}, a)$ is a space of \emph{relational
    predicates}.
\end{defn}

An application is a relational-predicate symbol paired with a sequence
of logical variables of length matching its arity.
To simplify the presentation, all objects defined in the remainder of
the section are defined over a fixed space of predicate symbols
$\mathcal{R}$ in which each symbol has arity $k \in \nats$, and a
fixed space of logical variables $X$.
\begin{defn}
  \label{defn:pred-apps}
  For relational predicate $R \in \mathcal{R}$ and sequence of variables
  $Y \in X^{*}$ such that $|Y| = k$, $(R, Y)$ is an \emph{application}
  of $\mathcal{R}$ over $X$.
\end{defn}
The space of applications of symbols in $\mathcal{R}$ over $X$ is
denoted $\apps{\mathcal{R}}{X}$.
For each application $A \in \apps{\mathcal{R}}{X}$, the predicate
symbol and argument sequence of $A$ are denoted $\relof{A}$ and
$\argsof{A}$ respectively.
A clause is a set of applications, a constraint over logical
variables, and a head relational predicate.
\begin{defn}
  \label{defn:chcs-structure}
  For $B \in \apps{\mathcal{R}}{X}^{*}$, %
  $\varphi \in \formulas{X}$, and %
  $H \in \mathcal{R}$, $(B, \varphi, H)$ is a \emph{Constrained Horn
    Clause}.
\end{defn}
The space of Constrained Horn Clauses over $\mathcal{R}$ and $X$ is
denoted $\chc{ \mathcal{R} }{ X }$.
For each $C \in \chc{\mathcal{R}}{X}$, the body of applications,
constraint, and head of $C$ are denoted $\appsof{C}$, $\ctrof{C}$, and
$\headof{C}$, respectively.
A set of Constrained Horn Clauses and a \emph{query}
relational-predicate symbol is a \emph{system} of Constrained Horn
Clauses;
i.e., the space of systems of Constrained Horn Clauses is denoted
$\chcs{\mathcal{R}}{X} = \pset(\chc{\mathcal{R}}{X}) \times
\mathcal{R}$.

\paragraph{Models}
A \emph{model} of $\mathcal{S}$ is a collection of logical models that
certify that $\mathcal{S}$ does not have a solution.
A derivation of a CHC system $\mathcal{S}$ is a tree $D$ labeled with
relational predicates of $\mathcal{S}$ such that all children with a
common parent $p$ are labeled with relational symbols of applications
in the body of a common clause of $\mathcal{S}$ with $p$ as its head.
For $\mathcal{C} \subseteq \chc{ \mathcal{R} }{ X }$, let each
relational predicate that is not the head of any clause of
$\mathcal{C}$ be a \emph{ground} relational predicate of
$\mathcal{C}$.
\begin{defn}
  \label{defn:chc-derivation}
  For $\mathcal{C} \subseteq \chc{ \mathcal{R} }{ X }$ and %
  $Q \in \mathcal{R}$, %
  $E \subseteq N^{*} \times N \times \mathcal{C}$, and %
  $\lambda_E : E \to \mathcal{C}$ be such that %
  $(N, E)$ is an directed hypertree and there is some labeling
  function $\lambda_N : N \to \mathcal{R}$ such that:
  \textbf{(1)} for $r \in N$ the root of $(N, E)$, %
  $\lambda_N(r) = Q$;
  \textbf{(2)} for each $n \in N$ a leaf of $(N, E)$, %
  $\lambda_N(n)$ is a ground relational predicate of $\mathcal{C}$;
  \textbf{(3)} for all $n_0, \ldots, n_k, n \in N$ with $e = ([ n_0,
  \ldots, n_k ], n) \in E$, it holds that %
  \textbf{(a)} $\lambda_N(n) = \headof{ \lambda_E(e) }$ and %
  \textbf{(b)} for each $0 \leq i \leq k$, %
  $\lambda_N(n_i) = \relof{ \appsof{ C }_i }$.
  Then $(N, E, \lambda_E)$ is a \emph{derivation} of $(\mathcal{C},
  Q)$.
\end{defn}
For $\mathcal{S} \in \chcs{ \mathcal{R} }{ X }$, the space of
derivations of $\mathcal{S}$ is denoted $\derivations{ \mathcal{S} }$.
For each $D \in \derivations{ \mathcal{S} }$, the nodes and hyperedges
of $D$ are denoted $\nodesof{D}$ and $\hedgesof{D}$.
%

A model of a CHC system $\mathcal{S}$ is a derivation $D$ of
$\mathcal{S}$ and a model indexed on nodes of $D$ that satisfies the
clauses that label the edges of $D$.
\begin{defn}
  \label{defn:chc-model}
  For $\mathcal{S} \in \chcs{\mathcal{R}}{X}$, %
  finite space $N$, %
  $E \subseteq N^{*} \times N$, %
  $\lambda_E : E \to \clausesof{ \mathcal{S} }$ such that $(N, E,
  \lambda_E) \in \derivations{ \mathcal{S} }$, %
  let $m$ be an \euflia model and %
  let $i : N \to \modelsof{ X }$ be such that %
  for all $n, n_0, \ldots, n_k \in N$ and %
  each $C \in \clausesof{S}$ with $e = ([ n_0, \ldots, n_k ], n) \in
  E$, it holds that for $C = \lambda_E(e)$,
  \textbf{(1)} $m, i(n) \sats \ctrof{ C }$;
  \textbf{(2)} for each $0 \leq j \leq k$, %
  $i(m)(\argsof{ \appsof{C}_j }) = i(n_j)(\params)$.
  Then $(D, m, i)$ is a \emph{model} of $\mathcal{S}$.
\end{defn}
We denote the models of $\mathcal{S}$ as $\modelsof{ \mathcal{S} }$.
For each $D \in \derivations{ \mathcal{S} }$, if there is some \euflia
model $m$ and $i : \nodesof{D} \to \modelsof{X}$ such that $(D, m, i)$
is a model of $\mathcal{S}$, then $D$ is \emph{feasible};
otherwise, $D$ is infeasible.
If some derivation of $\mathcal{S}$ is feasible, then $\mathcal{S}$ is
feasible (otherwise, $\mathcal{S}$ is infeasible).

A CHC \emph{solver} is a procedure that, given CHC system
$\mathcal{S}$, returns either the value $\isempty$ to denote that
$\mathcal{S}$ is infeasible, or a model of $\mathcal{S}$;
several CHC solvers have been proposed in previous
work~\cite{bjorner13,rummer13}.
\sys uses a CHC solver, named \solvechc, as a black box.

\paragraph{Solutions}
A solution of a clause $C$ is an interpretation of relational
predicates such that the conjunction of interpretations of all
relational predicates in the body of $C$ and the constraint of $C$
entail the interpretation of the head of $C$.
A solution of a CHC system $\mathcal{S}$ is a solution of each clause
in $\mathcal{S}$ that interprets the query relational predicate of
$\mathcal{S}$ as an unsatisfiable formula.
\begin{defn}
  \label{defn:chcs-soln}
  For $\mathcal{C} \subseteq \chc{ \mathcal{R} }{ X }$ and $Q \in
  \mathcal{R}$, let %
  $i : \mathcal{R} \to \formulas{ \params }$ be such that under each
  \euflia model, %
  \textbf{(1)} for each $B \in \apps{\mathcal{R}}{X}^{*}$, %
  $H \in \mathcal{R}$, and %
  $\varphi \in \formulas{X}$ such that $(B, H, C) \in \mathcal{C}$,
  \[ \elts{ i(\relof{A})[ \varsof{A} ] }_{A \in B}, \varphi \entails %
  i(H)
  \]
  \textbf{(2)} $i(Q) \entails \false$.
  Then $i$ is a \emph{solution} of $(\mathcal{C}, Q)$.
\end{defn}
When a CHC solver determines that a given system is empty, the solver
can synthesize a solution that certifies emptiness.
If \sys runs \solvechc on a system $\mathcal{S}$ and \solvechc
determines that $\mathcal{S}$ is empty, then \sys does not require
\solvechc to provide the generated solution.
However, \sys could be adapted to generate the CHC system that it
synthesizes, accompanied by its solution found by \solvechc, as a
proof of safety that can be independently certified.


\section{Proofs as Run Grammars}
\label{sec:pf-structures}
In this section and \autoref{sec:verifier-algo}, we present \sys in
technical detail.
In this section, we give a class of proofs of program safety as
grammars of program runs annotated with invariants.

\sys, given program \cc{P}, attempts to determine whether or not
\cc{P} is safe by synthesizing a CHC system $\mathcal{S}$ such that %
\textbf{(1)} each run of \cc{P} corresponds to a model of
$\mathcal{S}$ and %
\textbf{(2)} $\mathcal{S}$ has no models.
\sys attempts to certify that $\mathcal{S}$ has no models by synthesizing
a solution of \cc{S} in the theory $\datauif$.
Let \cc{P} be a fixed, arbitrary program for the remainder of the
section.

Let the theory $\runbg$ be $\datauif$ restricted to contain the
following uninterpreted function symbols.
The nullary symbols of $\runbg$ contain, for each $\cc{L} \in \locs$,
the symbol \cc{L}.
%
The unary symbols of $\runbg$ contain %
symbols $\locsym$ and $\ctrlsucc$; %
for each $\cc{x} \in \datavars$, the symbol \cc{x};
for each $\cc{p} \in \objvars$, the symbol \cc{p}.
The binary symbols of $\runbg$ contain %
for each $\cc{f} \in \datafields$, the symbol \cc{f};
for each $\cc{g} \in \objfields$, the symbol \cc{g}.

If for each $m \in \modelsof{ \mathcal{S} }$ with domain $N$, the
graph $(N, \succ_m, \locsym_m)$ is a control path of \cc{P}, then
$\mathcal{S}$ is a \emph{run grammar} of \cc{P}.
The space of run grammars of \cc{P} is denoted $\runchcs{ \cc{P} }$.
Let $\mathcal{S} \in \runchcs{ \cc{P} }$ be a fixed, arbitrary element
for the remainder of the section.

\sys uses a procedure $\pathof$ that, given $m \in \modelsof{
  \mathcal{S} }$, returns the control path of the run of $m$.
$\pathof$ is implemented by returning the interpretations of symbols
$\locsym$ and $\succsym$ in $m$.

Each run $r$ defines a map for each data and object field \cc{f}, from
each control point $p$ of $r$ and object $o$ allocated while executing
$r$ to the control point in $r$ at which the \cc{f} field of $o$ was
last updated when $r$ reached $p$.
In particular, %
let $N$ be a finite space, %
$L : N \to \locs$; %
$\succ: N \to N$; %
$\sigma : N \to \states$ be such that $((N, \succ, \locsym), \sigma)
\in \runsof{ \cc{P} }$.
For each $\cc{x} \in \datavars$, $\cc{x}_r : N \to \values$ is such
that for each $n \in N$, $\cc{x}_r(n) = \sigma(n)(\cc{x})$, and for
each $\cc{p} \in \objvars$, $\cc{p}_r : N \to \objs$ is defined
similarly.
For each $\cc{f} \in \datafields$, let $cc{f}_r : N \times \objs \to
N$ be such that for all $n_i, n_j \in N$ and $o \in \objs$, if $n_i$
is the last node before $n_j$ such that $\instrat{ \cc{P} }{
  \locsym(n_{i - 1}) }{ n_i } \equiv \store{p}{f}{x}$ and
$\objframeof{ \sigma(n_i) }(\cc{x}) = o$, then $\cc{f}_r(n_j, o) =
n_i$.
For each $\cc{g} \in \objfields$, let $\cc{g}_r : N \times \objs \to
N$ be defined similarly.
Both collections of symbols are called \emph{update histories}.

For $\mathcal{S} \in \runchcs{ \cc{P} }$, let domain $N$ combined with
$L$, $\succ$, %
$\setformer{ \cc{x}_r }{ \cc{x} \in \datavars}$, %
$\setformer{ \cc{p}_r }{ \cc{p} \in \objvars}$, %
$\setformer{ \cc{f}_r }{ \cc{f} \in \datafields }$, %
$\setformer{ \cc{f}_r }{ \cc{f} \in \objfields }$ %
as interpretations of %
$\locsym$, $\succsym$, %
$\datavars$, %
$\objvars$, %
$\datafields$, and %
$\objfields$ be the model of $\runbg$ denoted $\modelof{r}$.
If $\modelof{r}$ is a model of $\mathcal{S}$, then $r$ is a \emph{run}
of $\mathcal{S}$.
If each run of \cc{P} is a run of $\mathcal{S}$, then \cc{P} is
\emph{simulated} by $\mathcal{S}$.

For $\cc{p} \in \pathsof{ \cc{P} }$, if for each $r \in \runsof{
  \cc{p} }$, $\modelof{r}$ is not a model of $\mathcal{S}$, then
\cc{p} is \emph{refuted} by $\mathcal{S}$.
For $\cc{Q} \subseteq \pathsof{ \cc{P} }$, if for each $\cc{p} \in
\cc{Q}$, it holds that \cc{p} is refuted by $\mathcal{S}$, then
$\cc{Q}$ are refuted by $\mathcal{S}$.
If $\pathsof{ \cc{P} }$ are refuted by $\mathcal{S}$, then \cc{P} is
refuted by $\mathcal{S}$.

If a program is simulated and refuted by a CHC system, then the
program is safe.
\begin{lemma}
  \label{lemma:pf-corr}
  If \cc{P} is simulated by $\mathcal{S}$ and \cc{P} is refuted by
  $\mathcal{S}$, then \cc{P} is safe (\autoref{sec:semantics}).
\end{lemma}
If \cc{P} is simulated by a run grammar $\mathcal{S}$, then a solution
of $\mathcal{S}$ (which certifies that $\mathcal{S}$ is infeasible and
thus refutes \cc{P}), can be viewed as \emph{relational invariants}
that prove the safety of \cc{P}.

\begin{ex}
  \label{ex:empty-sim}
  Recall $\gbi$, the grammar partially depicted in
  \autoref{sec:ex-grammar}, \autoref{fig:build-inspect-grammar}. If we encode
  the runs of \cc{buildInspect} using the structure of $\gbi$, we
  could simulate every run, and showing the emptiness of that grammar
  could show the safety of $\cc{buildInspect}$.
\end{ex}
However, a safe program \cc{P} may be simulated by a CHC system that
does not refute it.
\begin{ex}
  \label{ex:non-empty-sim}
  Recall $\gbi'$, the simplistic grammar proposed in
  \autoref{fig:build-inspect-cfg}. If we try to encode \cc{P} using
  this structure, we are compelled to overapproximate by dropping
  state information because an unbounded separation arises between a
  load and its matching store. Every runs of \cc{buildInspect} is
  present in $\gbi'$, but there are additional models not
  corresponding to actual runs. This non-empty grammar simulates
  \cc{buildInspect}, even though \cc{buildInspect} is safe.
\end{ex}

%

\section{Inductive synthesis of proofs using \sys}
\label{sec:verifier-algo}
In this section, we describe \sys, a verifier that attempts to prove
the safety of a given program by inductively synthesizing a run grammar
as proof.
\begin{figure}[t]
  \centering
  \begin{minipage}[t]{.48\textwidth}
\begin{algorithm}[H]
  \SetKwInOut{Input}{Input}
  \SetKwInOut{Output}{Output}
  \SetKwProg{myproc}{Procedure}{}{}
  \Input{$\cc{P} \in \lang$.}
  \Output{Decision as to whether \cc{P} is safe.}
  \myproc{$\sys(\cc{P})$ \label{line:sys-begin}} %
  { \myproc{$\sysaux(\feedback)$ \label{line:begin-aux}} %
    { $\mathcal{G} \assign \synchc(\cc{P}, \feedback)$ %
      \label{line:syn-run-grammar} \;
      \Switch{$\solvechc(\mathcal{G})$ \label{line:run-solvechc}} %
      { \uCase{$\isempty$: \label{line:is-empty} } %
        { %
          \Return{$\true$} \label{line:ret-safe} %
        } %
        \Case{$D \in \derivations{\mathcal{G}}$: \label{line:has-model}} %
        { 
          $p \assign \pathof(D)$ \label{line:ext-path} \; %
          \uIf{$\isfeasible(p)$ \label{line:test-feasible} }{ %
            \Return{$\false$} \label{line:ret-unsafe} \; } %
          \Else{ \Return{$\sysaux(\add{p}{\feedback})$} %
            \label{line:ret-recurse} \; }
        } %
      } \label{line:end-switch} %
    } \label{line:end-aux} %
    \Return{$\sysaux(\emptyset)$} \label{line:run-aux}
  }
\end{algorithm}
\end{minipage}
\begin{minipage}[t]{.46\textwidth}
\begin{algorithm}[H]
  \SetKwInOut{Input}{Input}
  \SetKwInOut{Output}{Output}
  \SetKwProg{myproc}{Procedure}{}{}
  \Input{$\cc{P} \in \lang$ and $p \in \pathsof{ \cc{P} }$.}
  \Output{A minimal refuting neighborhood $\nu : \nodesof{p} \to
    \pset(\nodesof{p})$ for $p$.}
  \myproc{$\extractdeps(\cc{P}, p)$} %
  { 
    $\nu \assign \nuAll$ \label{line:get-all-deps} \; %
    \For{$n, n' \in \nodesof{p}$} {
      {$\nu' \assign \upd{\nu}{n}{\nu(n) \setminus \{n'\}}$ %
        \label{line:prune-try}} \;
      \lIf{$\lnot \eufliaissat(\sympath(\cc{P},p,\nu'))$ %
          \label{line:prune-check} } {
          $\nu \assign \nu'$ \label{line:prune-commit}
      } %
    } %
    \Return{$\nu$} %
  } %
\end{algorithm}
\end{minipage}
  \begin{minipage}[t]{.48\textwidth}
    \begin{algorithm}[H]
      \caption{\sys: a safety verifier based on inductive synthesis. %
        \sys uses procedures %
        $\synchc$ (\autoref{sec:syn-skeleton}, \autoref{sec:syn-run-grammar}), %
        $\solvechc$ (\autoref{sec:chcs}), %
        $\pathof$ (\autoref{sec:pf-structures}), and %
        $\isfeasible$ (\autoref{sec:history-ctrs}). }
      \label{alg:sys}
    \end{algorithm}
  \end{minipage}
  \qquad
  \begin{minipage}[t]{.46\textwidth}
    \begin{algorithm}[H]
      \caption{$\extractdeps$: given $\cc{P} \in \lang$ and infeasible
        $p \in \pathsof{\cc{P}}$, returns a minimal refuting
        neighborhood for $p$. The map $\nuAll$ and the formula
        $\sympath$ are defined in \autoref{sec:history-ctrs}.}
      \label{alg:extract}
    \end{algorithm}
  \end{minipage}
\end{figure}

\autoref{alg:sys} contains pseudocode for \sys.
\sys defines a procedure \sysaux which, given infeasible paths
$\feedback \subseteq \pathsof{ \cc{P} }$, attempts to determine if
\cc{P} is safe by synthesizing a run grammar that overapproximates the
runs of \cc{P} and refutes $\feedback$
(\autoref{line:begin-aux}---\autoref{line:end-aux}).
\sys invokes \sysaux on the empty set of control paths and returns the
result (\autoref{line:run-aux}).

\sysaux, given $\feedback$, synthesizes a run grammar $\mathcal{G}$
that simulates \cc{P} and refutes $\feedback$ by a procedure \synchc
on \cc{P} and $\feedback$ (\autoref{line:syn-run-grammar}).
An implementation of \synchc that performs a reduction to constraint
solving is described in \autoref{sec:syn-run-grammar}.
\sysaux then determines if $\mathcal{G}$ refutes \cc{P} by running the
CHC solver \solvechc (described in \autoref{sec:chcs}) on
$\mathcal{G}$ to determine if $\mathcal{G}$ is infeasible.
If \solvechc determines that $\mathcal{G}$ is infeasible
(\autoref{line:is-empty}), then \sysaux returns that \cc{P} is safe.

Otherwise, \solvechc returns a feasible derivation $D$ of
$\mathcal{G}$ (\autoref{line:has-model}).
\sysaux extracts from $D$ some $\cc{p} \in \pathsof{ \cc{P} }$ not
refuted by $\mathcal{G}$ (\autoref{line:ext-path}).
\sysaux then tests if \cc{p} is feasible by running a procedure
\isfeasible on \cc{p} (\autoref{line:test-feasible});
if \isfeasible determines that \cc{p} is feasible, then \sysaux
returns that \cc{P} is not safe.

Otherwise, if \isfeasible returns that \cc{p} is feasible, then
\sysaux recurses on $\feedback$ extended with \cc{p}, and returns the
result of the recursion (\autoref{line:ret-recurse}).

We now describe the implementation of each procedure used by \sys.
In \autoref{sec:is-feasible}, we describe an implementation of
\isfeasible.
In \autoref{sec:syn-skeleton} and \autoref{sec:syn-run-grammar}, we
describe the two steps performed by \synchc.

\subsection{Testing path feasibility using \isfeasible}
\label{sec:is-feasible}

The procedure \isfeasible, given $\cc{P} \in \lang$ and $\cc{p} \in
\pathsof{ \cc{P} }$, returns whether or not \cc{p} is a feasible path
of \cc{P}.
Let finite space $N$, %
$E \subseteq N \times N$, and %
$\lambda : N \to \locs$, be such that $\cc{p} = (N, E, \lambda)$.
\isfeasible generates a constraint $\varphi \in \formulas{\runbg}$ for
which each model corresponds to run of \cc{p}.

\subsubsection{Symbolic constraints over update histories}
\label{sec:history-ctrs}
$\varphi$ is constructed as a conjunction of clauses, each of which
models the effect of an instruction executed in a step of \cc{p} on
its local variables and update histories.
The effect of each $\cc{i} \in \cc{P}$ executed in a step of \cc{p} is
formulated by a constraint $\symrel{ \cc{i} }$ parameterized on $N$,
along with distinguished $n, n' \in N$ that model state before
executing \cc{i} and state that immediately results from executing
\cc{i}.

$\symrel{ \cc{i} }(n, n', N)$ is defined casewise by the structure of
\cc{i}.
In many cases, $\symrel{ \cc{i} }$ is defined using formulas that
constrain equality of logical terms that model state.
In particular, let %
$\datavarseq{n}{n'}$, $\objvarseq{n}{n'}$, $\datafieldseq{n}{n'}{N}$,
and $\objfieldseq{n}{n'}{N}$ constrain that the states at $q$ and $n'$
have equal local value states, local object states, data timestamps,
and object timestamps, respectively.
I.e.,
\begin{align*}
  \datavarseq{n}{n'} & \equiv \bigland_{ \cc{x} \in \datavars} %
  \cc{x}(n) = \cc{x}(n') \\
  \datafieldseq{n}{n'}{N} & \equiv %
  \bigland_{ %
    \substack{ %
      \cc{f} \in \datafields \\
      \cc{p} \in \objvars \\
      n'' \in N } } %
  \cc{f}(n, \cc{p}(n'')) = \cc{f}(n', \cc{p}(n''))
\end{align*}
$\objvarseq{n}{n'}$ and $\objfieldseq{n}{n'}{N}$ are defined
similarly.

If $\cc{i} \in \valinstrs$, then $\symrel{\cc{i}}$ constrains that the
value local state at $n$ and the value local state at $n'$ are in the
transition relation of \cc{i}, and that their object stores are
identical.
I.e., $\symrel{ \cc{i} }(n, n', N)$ is
\begin{align*}
  \symvaltrans{ \cc{i} }[ \datavars(n), \datavars(n') ] \land
  \objvarseq{ n }{ n' } \land %
  \nonumber \\
  \datafieldseq{ n }{ n' }{N} \land %
  \objfieldseq{ n }{ n' }{N}
\end{align*}
For each $\cc{i} \in \niltests$ or $\cc{i} \in \objeqs$, the
constraint $\symrel{ \cc{i} }$ is defined similarly.

For $\cc{x} \in \valvars$, $\cc{p} \in \objvars$, and $\cc{f} \in
\datafields$, $\symrel{ \store{p}{f}{x} }(n, n', N)$ constrains that
at $n'$, the most recent store to the \cc{f} field of the object bound
to \cc{p} is $n'$.
The local states and update history of all other fields are identical
between $n$ and $n'$.
I.e., $\symrel{ \store{p}{f}{x} }(n, n', N)$ is
\begin{align*}
  & \datavarseq{ n }{ n' } \land %
  \objvarseq{ n }{ n' } \land %
  \cc{f}(n', \cc{p}(n')) = n' \land
  \bigland_{ %
    \substack{
      \cc{p} \in \objvars \\
      n'' \in N }} %
  \cc{f}(n', \cc{p}(n'')) = \cc{f}(n, \cc{p}(n'')) \land \\ %
  & \bigland_{ %
    \substack{ %
      \cc{g} \not= \cc{f} \in \datafields \\
      \cc{p} \in \objvars \\
      n'' \in N } } \cc{g}(n', \cc{p}(n'')) = \cc{g}(n, \cc{p}(n''))
  \land \objfieldseq{n}{n'}{N}
\end{align*}
For each $\cc{i} \in \objstores$, the constraint $\symrel{ \cc{i} }$
is defined similarly.
\begin{ex}
  \label{ex:store-ctr}
  \cc{buildInspect} contains an object store on line \cc{7}. The
  relation which models this instruction constrains the local state
  after the instruction such that the most recent store for the
  \cc{next} field of the object \cc{tail} is set to \cc{tmp}.
\end{ex}

For each $\cc{x} \in \datavars, \cc{p} \in \objvars$ and $\cc{f} \in
\datafields$, $\symrel{\load{x}{p}{f}}(n, n', N)$ inspects all states
bound to variables in $Q$ to determine if some $n'' \in N$ is the
point of the most recent store to \cc{f}.
If so, the value in \cc{x} at $n'$ is constrained to be the value
stored when stepping to $n'$;
otherwise, the value bound to \cc{x} at $n'$ is unconstrained.
Let $\datastoresrcs{ \cc{P} } \subseteq \locs$ be control locations that
are sources of data stores in \cc{P}, and let $\stored{ \cc{P} } :
\datastoresrcs{ \cc{P} } \to \datavars$ map each such control location
to the data variable that holds that value stored by the instruction.
Then $\symrel{ \load{x}{p}{f} }(n, n', N)$ is
\begin{align*}
  \bigland_{ %
    n'' \in N, \lambda(n'') \in \datastoresrcs{ \cc{P} } } & %
  \cc{f}(q, \cc{p}(q)) = n'' \implies %
  \cc{x}(q') = \stored{\cc{P}}(n'') \land \\
  \bigland_{\cc{y} \not= \cc{x} \in \datavars} & \cc{y}(n') = \cc{y}(n)
  \land
  \objvarseq{n}{n'} \land %
  \datafieldseq{n}{n'}{N} \land %
  \objfieldseq{n}{n'}{N}
\end{align*}
For each $\cc{i} \in \objloads$, the constraint $\symrel{ \cc{i} }$ is
defined similarly.
\begin{ex}
  \label{ex:load-ctr}
  \cc{buildInspect} contains an object load on line \cc{11}. The
  relation which models this instruction constrains the local state
  after the instruction such that the variable \cc{elt} is set to the
  most recent store for the \cc{next} field of the object \cc{elt}.
\end{ex}

For $\cc{x} \in \objvars$, $\symrel{\alloc{x}}(n, n', N)$ constrains
that the identity of the allocated object is $n'$;
the fields of the allocated object are initialized at $n'$.
\begin{align*}
  &\datavarseq{ n }{ n' } \land %
  \cc{x}(n') = n' \land %
  \bigland_{ \cc{y} \not= \cc{x} \in \objvars } %
  \cc{y}(n') = \cc{y}(n) \land %
  \cc{Stored}(n') = 0 \land \\ %
  &\bigland_{ \cc{f} \in \datafields \union \objfields } {
    \cc{f}(n', \cc{x}(n')) = n'
  } \land %
  \bigland_{
    \substack{
      \cc{y} \not= \cc{x} \in \objvars \\
      \cc{f} \in \datafields \union \objfields }
  } %
  \cc{f}(n', \cc{y}(n')) = \cc{f}(n, \cc{y}(n))
\end{align*}

For a path $p = (N, \lambda, E) \in \pathsof{\cc{P}}$, define the
constraint $\sympath(\cc{P},p,\nu) \in \formulas \runbg$, where $\nu :
N \to \pset(N)$, as follows:
\[ 
  \bigland_{ (n, n') \in E } %
  \symrel{ \instrat{ \cc{P} }{ \lambda(n) }{ \lambda(n') } }%
  (n, n', \nu(n'))
\]

Let $\nuAll : N \to \pset(N)$ be such that for each $n \in N$,
$\nuAll(n) = N$.
\isfeasible generates the constraint $\varphi = \sympath(\cc{P}, p,
\nuAll)$.
\isfeasible returns that \cc{p} is a feasible path of \cc{P} if and
only if \eufliaissat (see \autoref{sec:chcs}) decides that $\varphi$
is satisfiable.

\isfeasible is a sound and complete procedure for testing path
feasibility.
\begin{lemma}
  \label{lemma:is-feasible-corr}
  If \cc{p} is a feasible path of \cc{P}, then $\isfeasible(\cc{P},
  \cc{p}) = \true$.
  Otherwise, $\isfeasible(\cc{P}, \cc{p}) = \false$.
\end{lemma}

\isfeasible could be implemented alternatively by reduction to
satisfiability testing in alternative theories that can soundly and
completely model the feasibility of bounded paths, such as
combinations of the theory of arrays.
We have presented \isfeasible as using \eufliaissat in order to
introduce the symbolic relations on update histories for each
instruction, which \sys also uses in queries in order to synthesize
run grammar skeletons (\autoref{sec:syn-skeleton}), and as components
of run grammars (\autoref{sec:syn-run-grammar}).
  
\subsection{Synthesizing a run-grammar skeleton using constraint
  solving}
\label{sec:syn-skeleton}
In this section and in \autoref{sec:syn-run-grammar}, we describe an
implementation of \synchc, which given $\cc{P} \in \lang$ and $\cc{F}
\subseteq \pathsof{P}$, synthesizes a run grammar $\mathcal{G}$ that
simulates \cc{P} and refutes \cc{F}.
\synchc synthesizes such a run grammar by performing two steps,
described in this section and \autoref{sec:syn-run-grammar}.
In the first step, \synchc finds the relations of $\mathcal{G}$ and
the sets of variables in each clause body to which relational
predicates are applied.
We refer to such an object as a CHC system \emph{skeleton}.
\synchc finds a skeleton by reduction to constraint solving.

\synchc obtains a clause skeleton for \cc{P} and \cc{F} by running a
procedure \synskeleton on \cc{P} and \cc{F}.
The implementation of \synskeleton that we present always makes
progress, in the sense that for each \cc{P} and \cc{F}, it synthesizes
a well-formed skeleton for \cc{P} and \cc{F}.
However, the implementation is not complete, in that there are some
safe programs whose required skeletons are not expressible in this
particular constraint-based approach to \synskeleton.
We have found that the present implementation is sufficiently
expressive for \sys to be able to prove correctness of interesting
verification challenge problems (discussed in
\autoref{sec:evaluation}).
We leave the design of alternative, more expressive, or more efficient
implementations of \synskeleton for future work.

\subsubsection{Run-grammar skeletons}
\label{sec:skeletons}
In its first step, \synchc synthesizes a run-grammar \emph{skeleton},
which defines the control paths generated by run grammar constructed
from completing it.
We define a restricted implementation of \synchc that only generates
run grammars in which each non-terminal derives exactly two control
subpaths of \cc{P}.
In particular, let $\skelpreds$ be a set of relational predicates of
fixed, common arity $n$.
Let $Q$ be a space of variable symbols, and let $Q' \in Q^{*}$ be
fixed sequence of $n$ distinct variables in $Q$.
For %
$R \in \skelpreds$,
$A \in \apps{ \skelpreds }{ Q }$, %
$L : Q \to \locs$, %
$q, q' \in Q$, $\mathcal{C} = (R, A, L, q, q')$ is a \emph{clause
  skeleton}.
The space of clause skeletons is denoted $\clauseskels$.
A set of clause skeletons combined with four maps from $\skelpreds \to
\ints_n$ is a run-grammar skeleton.
For the remainder of this section, let $R$, $A$, $L$, $q$, $q'$, and
$\mathcal{S} \subseteq \clauseskels$ and $F = \prefrag_0,
\postfrag_0, \prefrag_1, \postfrag_1 : \skelpreds \to \ints_n$ be a
fixed, arbitrary elements.

\paragraph{Synthesizing a skeleton that simulates a given program}
$(\mathcal{S}, F)$ simulate \cc{P} if when each clause in
$\mathcal{S}$ is extended with a suitable constraint to form a clause,
the resulting run grammar generates all control paths of \cc{P}.
In particular, for $\varphi \in \formulas{\runbg}$, let $((A,
\varphi), R)$ be the \emph{completion} of $\mathcal{C}$.
Let $\mathcal{S}' \in \runchcs$ such that each clause in
$\mathcal{S}'$ is a completion of a clause $(R, A, L, q, q')$ with a
constraint $\bigland_{q \in Q} \locsym(q) = L(q) \land q' =
\succsym(q)$ be the \emph{control-path grammar} of $\mathcal{S}$.
If for each $\mathcal{R}' \in \skelpreds$ and each model $m$ of
$\mathcal{R}$, there is a control path from
$m(\prefrag_0(\mathcal{R}'))$ to $m(\postfrag_0(\mathcal{R}'))$ and a
control path from $m(\prefrag_1(\mathcal{R}'))$ to
$m(\postfrag_1(\mathcal{R}'))$, then $(\mathcal{S}, F)$ is
\emph{well-formed}.
If $\mathcal{S}'$ simulates \cc{p}, then $\mathcal{S}$ simulates
\cc{P}.
If all $r, r' \in \runsof{ \mathcal{S}' }$ that have the same control
path have the same derivation in $\mathcal{S}'$, then $\mathcal{S}$ is
\emph{unambiguous}.

\synskeleton synthesizes a well-formed, unambiguous control-path
grammar that simulates \cc{P} by reduction to constraint solving.
In particular, \sys generates an \euflia constraint $\varphi_{ \cc{P}
}$ such that each \emph{model} of $\varphi_{ \cc{P} }$ defines a set
of \emph{clause} skeletons that simulate \cc{P}.
For each clause skeleton $C$, the problem of choosing a relational
predicate to apply in the body of $C$, the set of variables to which
it is applied, the map from variables to locations, and instruction
source and destination variables can be encoded directly as an \euflia
constraint.

\paragraph{Synthesizing a skeleton that refutes control paths}
$\mathcal{C}$ refutes \cc{F} if for each $\cc{p} \in F$, the
derivations of the control extension of $\mathcal{C}$ simultaneously
derive sufficient sets of matching loads and stores in $\cc{p}$ that a
suitable extension of $\mathcal{C}$ (described in
\autoref{sec:syn-run-grammar}) refutes \cc{p}.
Such sufficient sets are formulated precisely as a minimal refuting
neighborhood of \cc{p}.

For $\cc{p} \in \cc{F}$ with $N = \nodesof{ \cc{p} }$ and $\nu : N \to
\pset(N)$ such that $\sympath[ \cc{P}, p, \nu ]$ is unsatisfiable,
$\nu$ is a \emph{refuting neighborhood} of $\cc{p}$.
For all $\nu, \nu' : N \to \pset$, if for all $n \in N$, $\nu'(n)
\subseteq \nu(n')$, then $\nu'$ is \emph{contained} by $\nu$;
if, in addition, $\nu$ is not contained by $\nu'$, then $\nu'$ is
\emph{strictly} contained by $\nu$.
If for each $\nu' : N \to \pset(N)$ that is strictly contained by
$\nu$, it holds that $\nu'$ is not refuting neighborhood of \cc{p},
then $\nu$ is a \emph{minimal} refuting neighborhood of \cc{p}.

\extractdeps (\autoref{alg:extract}), given \cc{P} and an infeasible
$\cc{p} \in \pathsof{\cc{P}}$, returns a minimal refuting neighborhood
of $\cc{p}$.
\extractdeps maintains a refuting neighborhood of \cc{p}, from which
it iteratively minimizes entries until it obtains a minimal refuting
neighborhood.
\extractdeps constructs $\nuAll$ as an initial refuting neighborhood
(\autoref{line:get-all-deps}).
For all $n, n' \in N$, it determines if the map $\nu'$ obtained by
removing $n'$ from the image of $n$ in its maintained refuting
neighborhood $\nu$ (\autoref{line:prune-try}) is a refuting
neighborhood.
To determine this fact, \extractdeps runs \eufliaissat on $\sympath[
\cc{P}, \cc{p}, \nu' ]$ (\autoref{line:prune-check}).
If $\nu'$ is a refuting neighborhood, then \extractdeps updates its
maintained refuting neighborhood to be $\nu'$
(\autoref{line:prune-commit}).

For each $\cc{p} \in \cc{F}$ with $N = \nodesof{ \cc{p}}$,
\synskeleton synthesizes a minimal refuting neighborhood of \cc{p},
named $\nu_{ \cc{p} } : N \to \pset(N)$ by running \extractdeps on
\cc{P} and \cc{p}.
\synskeleton then constructs a constraint $\varphi_{ \cc{p} }$ in
which each solution defines \textbf{(1)} a set of clause skeletons,
over the vocabulary interpreted as a clause skeleton in each solution
of $\varphi_{\cc{P}}$, and \textbf{(2)} a derivation $D$ of \cc{p} in
which for each $n \in N$ and each $n' \in \nu_{ \cc{p} }(n)$, $n$ and
$n'$ are derived in a common instance of a clause in $D$.

\begin{ex}
  \label{ex:refutes-feedback}
  Recall once again $\gbi'$, the simplistic grammar of
  \autoref{fig:build-inspect-cfg}. As a skeleton, $\gbi'$ fails to
  refute even a short infeasible path of \cc{buildInspect}. For
  example, if a path performs a store and a load (the transitions from
  line \cc{7} to \cc{8} and from line \cc{11} to \cc{10}), a refuting
  neighborhood must put the state at line \cc{10} into the
  neighborhood of the state at line \cc{8} or vice versa. One way to
  do this is to modify $\gbi'$ to include a special case rule for this
  path. A better way is to construct the grammar $\gbi$ from
  \autoref{fig:build-inspect-grammar}, which refutes much large set of
  infeasible paths. In fact, this skeleton refutes all paths in
  \cc{buildInspect}.
\end{ex}

We have described a particular implementation of \synskeleton that we
have implemented in the current version of \sys.
The current of implementation of \synskeleton is restricted, in that
it only synthesizes skeletons of run grammars that are linear and
unambiguous.
\synskeleton finds such skeletons by reduction to constraint solving.
The key motivation for adopting such limitations was that a
relatively simple but useful version of \synskeleton could be designed
by reusing existing, heavily-optimized algorithms implemented in
constraint solvers.
However, \sys only requires that an implementation of \synskeleton,
given program \cc{P} and infeasible paths \cc{F}, synthesize a
skeleton encodes all paths of \cc{P}, including refuting neighborhoods
for the paths in\cc{F}.
Alternative, explicit implementations of \synskeleton may yield
significant benefits over the current version based on constraint
solving.

\subsection{Synthesizing a run grammar from a skeleton}
\label{sec:syn-run-grammar}
In its second step, \synchc completes $\mathcal{C} \subseteq
\skeletons$, a set of clause skeletons that simulates \cc{P} and a
refutes \cc{F} synthesized by \synskeleton, to generate a run grammar
$\mathcal{G} \in \runchcs$ that simulates \cc{P} and refutes \cc{F}.
Let 
$R \in \skelpreds$, %
$L : Q \to \locs$, %
$q, q \in Q$ be such that $(R, L, q, q') \in \skeletons$.
From $(R, A, L, q, q')$, \synchc generates $\varphi \in
\formulas{ \runbg }$, and includes in $\mathcal{G}$ the clause
$\mathcal{C} = (A, \varphi, R)$.

$\varphi$ is a conjunction of constraints that model the effect of
\cc{i}, and effects of allocations and stores performed in steps of
execution other than \cc{i} (\autoref{sec:footprint}).
The first conjunct $\varphi_0$ constrains that each instance of $q'$
is a control success of each instance of $q$;
i.e., $\varphi_0 \equiv q' = \succsym(q)$.
The second conjunct $\varphi_1$ constrains that for path points $n$
and $n'$ bound to $q$ and $q'$ at an instance of $\mathcal{C}'$, the
state at $n'$ is the result of transitioning from the state at $n$,
under the instruction that connections $n$ and $n'$.
I.e., $\varphi_1 \equiv \symrel{ \instrat{ \cc{P} }{ L(n) }{ L(n') }
}(q, q', Q)$.

The third conjunct $\varphi_2$ models the effect of all instructions
that connect points derived by other clause instances on the objects
in scope when the step of $\mathcal{C}$ executes.
The construction of the constraint is described in detail below.

\subsubsection{Formulating the effect of instructions derived outside
  of a clause instance}
\label{sec:footprint}
The space of valid runs with update histories is partially constrained
by the symbolic relations for each instruction, defined in
\autoref{sec:history-ctrs}.
For each $\cc{i} \in \instrs$, $\symrel{\cc{i}}$ models the effect of
executing \cc{i} on a finite set of states.
When $\symrel{\cc{i}}$ is used to determine feasibility of an entire
path, in which case the set consists of all states in the path.
When $\symrel{\cc{i}}$ is used as a conjunct of a clause constraint,
the set of states consists of all states that are derived in the same
instance of a clause $\mathcal{C}$ that derives the states connected
by \cc{i}.
However, it does not include the set of all states in the derived
path, namely states that occur exclusively in subderivations of
\cc{C}, or that occur in a subderivation outside of the derivation of
$\mathcal{C}$ and are not provided as arguments to $\mathcal{C}$.

$\varphi_3$ constrains the effect of instructions executed in steps
from such states on objects that may only be bound to object variables
in states derived in an application of $\mathcal{C}$.
If $\cc{i}$ is not an object load or allocation, then each object in
scope bound to a variable in an application of $\mathcal{C}$ is in
scope for its child or and parent;
as a result, $\varphi \equiv \true$.

Otherwise, let $\cc{p} \in \objvars$ be the object variable bound by
\cc{i}.
Then $\varphi_3$ is a conjunction of two constraints.
Let $\mathcal{R}_A \in \skelpreds$ and $Q_A \in Q^{*}$ be such that $A
= \mathcal{R}_A(Q_A)$.
The first constraint, $\varphi_3^0 \in \formulas{\runbg}$, constrains
that if the object $o$ bound to $\cc{p}$ at $q$ is not in the scope of
any states bound to $Q_A$, then the update histories of $o$ is
identical across both control subpaths derived by the derivation with
head $A$.
I.e., $\varphi_3^0$ is
\begin{align*}
  \bigland_{ %
    \substack{ %
      q'' \in Q_A \\
      \cc{p}' \in \objvars } } & \cc{p}(q) \not= \cc{p}'(q'') \implies
  \\ %
  \bigland_{ %
    \cc{f} \in \objfields \union \datafields
  } & %
  \cc{f}(Q_A[ \prefrag_0(\mathcal{R}_A) ], \cc{p}(q)) = %
  \cc{f}( Q_A[ \postfrag_0(\mathcal{R}_A)], \cc{p}(q)) \land \\
  & \cc{f}(Q_A[ \prefrag_1(\mathcal{R}_A) ], \cc{p}(q)) = %
  \cc{f}( Q_A[ \postfrag_1(\mathcal{R}_A) ], \cc{p}(q)) 
\end{align*}
The second constraint, $\varphi_3^1 \in \formulas{\runbg}$, constrains
that if $o$ is not in an argument of $\mathcal{R}$ when $\mathcal{C}$
is applied, then the update histories of $o$ is identical between the
end of the first control subpath that it derives and the beginning of
the second control subpath that it derives.
I.e., $\varphi_3^1$ is
\begin{align*}
  \bigland_{ %
    \substack{
      q'' \in Q' \\
      \cc{p}' \in \objvars } } & \cc{p}(q) \not= \cc{p}'(q'') \implies
  \\ %
  \bigland_{ \cc{f} \in \objfields \union \datafields } & %
  \cc{f}(Q'[ \postfrag_0(\mathcal{R} ) ], \cc{p}(q)) = %
  \cc{f}(Q'[ \prefrag_1(\mathcal{R}) ], \cc{p}(q))
\end{align*}

\subsubsection{Key properties}
\label{sec:syn-grammar-props}
The partial correctness of \sys depends only on the fact that \synchc,
given a program and set of infeasible paths, synthesizes a run grammar
that simulates the program.
\begin{lemma}
  \label{lemma:synchc-corr}
   \cc{P} is simulated by \synchc(\cc{P}, \feedback).
\end{lemma}
\synchc also synthesizes a run grammar that refutes \cc{F}.
However, this fact is primarily useful for proving that \sysaux
progresses, in that it never collects the same control path from
distinct invocations of \solvechc.

\subsection{Key features}
\label{sec:features}
\sys is correct on all programs on which it terminates.
\begin{thm}
  \label{thm:soundness}
  If $\sys(\cc{P}) = \true$, then \cc{P} is safe, and if $\sys(\cc{P})
  = \false$, then \cc{P} is not safe.
\end{thm}
Because the problem of verifying safety of programs in $\lang$ is
undecidable, there are some programs on which $\sys$ will not
terminate.

One limitation of \sys as presented above is that it can only
effectively determine the safety of programs that contain all stores
of objects that they load.
Defining a logic of program summaries that \sys can both use and
validate is a conceptually challenging and critical direction for
future work.
\sys is motivated by a number of practical applications which it can
be applied to. However, having a program summary logic would enable a
direct and formal comparison on the theoretical level to other
approaches, particularly separation logic~\cite{reynolds02} and
effectively propositional reasoning~\cite{itzhaky13}.

\subsection{Discussion}

The prototype of \synskeleton, developed in
\autoref{sec:syn-skeleton}, has some shortcomings which make it
worthwhile to pursue a replacement. It is deficient in the
expressivity of the grammars it can return, and it is intractable on
modest programs. For an infeasible path $p$ in the feedback set
$\feedback$, the existing constraint-based approach cannot enforce
that every derivation of $p$ include a refuting neighborhood $\nu$ for
$p$, merely that one such derivation does; hence, the expressivity of
\synskeleton must be artificially limited to unambiguous grammars. If
it were possible to additionally supply negative examples to
\synskeleton (graphs of $p$ not including any refuting neighborhood),
then ambiguity would be permissible.

\newcommand{\lstar}{$\mathsf{L}^*\,$}
Learning a linear string grammar (a regular expression) through
queries to an oracle is a solved problem due to Angluin's \lstar
algorithm~\cite{angluin87}. To be precise, one supposes a teaching
oracle which can answer membership queries for a secret language $L$
and which can confirm or deny with a counterexample that a proposed
DFA $D$ has the property $L(D) = L$. \lstar is an efficient algorithm
for using this oracle to learn such a $D$ through both positive and
negative examples. \lstar is guaranteed to terminate when $L$ is
regular. The parallel between this problem statement and the task of
\synskeleton in \sys is striking. In \sys, a grammar $\mathcal{C}$ for
an unknown language of graphs is sought. $\isfeasible(\cc{P},p,\nu)$
answers the query of whether the control path graph $p =
(N,\lambda,E)$ unioned with the data edges $\dataedges$ implied by
$\nu$ is a member of this language. When a run grammar $\mathcal{G}$
is proposed, $\solvechc(\synchc(\cc{P},F))$ attempts to confirm or
deny with a counterexample that $\mathcal{G}$ correctly describes the
language.

The task of \synskeleton is more difficult than the regular language
learning task of \lstar because the unknown language is a set of
graphs, not strings. However, there is work towards \lstar-like
algorithms for richer classes of languages such as learners of
context-free grammars~\cite{angluin87,clark10,yoshinaka11} and
\emph{multiple context-free (MCF) grammars}~\cite{yoshinaka12}. A
single non-terminal symbol in an MCF grammar is not a hole in a single
string, but a vector of $r$ holes, where the \emph{rank} $r$ is
bounded. The rule for a non-terminal of rank $r$ can invoke other
non-terminals of rank $r$ by producing characters at the holes, of
rank less than $r$ by filling a hole, and of rank greater than $r$ by
splitting a hole into two adjacent holes. Again, a striking parallel
with \sys arises: this behavior of MCF non-terminals is exactly the
behavior of skeleton grammars returned by \synskeleton, where control
edges play the role of characters. It remains an open question to
determine if an MCF learner can be adapted to include the data edges
of a refuting neighborhood $\nu$ for each positive example $p$.


\section{Evaluation}
\label{sec:evaluation}
We empirically evaluated \sys in order to answer the following
questions:
\textbf{(1)} Can \sys verify the safety of low-level programs that
operate on unbounded data structures where existing approaches fail?
\textbf{(2)} Can \sys verify the safety of such programs efficiently?

We implemented \sys as a verifier for programs represented in Java
Virtual Machine (JVM) bytecode.
The set of benchmarks consists of some programs adapted from challenge
benchmarks in the SV-COMP benchmark collection~\cite{svcomp17},
a program that cannot be verified by a competing
approach~\cite{itzhaky13a},
novel benchmarks designed to exhibit particular capabilities of \sys,
and a program which demonstrates that \sys is limited to properties
that can be expressed in a context free way.

In short, the results of our experiments indicate that \sys can be
applied to verify the safety of programs that cannot be verified by
existing shape analyzers or verifiers.
The time required for the \sys implementation to perform verification
varies dramatically, and we diagnose what causes intractability in
some cases.

\autoref{sec:exp-procedure} describes our implementation of \sys for
JVM and our experimental procedure;
\autoref{sec:benchmarks} summarizes each benchmark;
\autoref{sec:ex-benchmarks} provide a qualitative evaluation of \sys
by describing the proofs that it synthesizes in order to prove the
safety of several illustrative benchmarks.
\autoref{sec:results} provides a quantitative evaluation of \sys
by analyzing its performance.

\subsection{Implementation and experimental procedure}
\label{sec:exp-procedure}
We implemented \sys as a verifier for JVM bytecode.
\sys can thus be applied to verify programs written in Java, Scala, or
other languages with compilers that target the JVM.
The hypothetical language, \lang, targeted by \sys was defined in
\autoref{sec:lang}, including classes of instructions that perform
operations over objects, such as loads, stores, and allocations.
Each such instruction directly corresponds to an instruction in the
JVM instruction set.
In addition, the implementation of \sys can be applied to programs
that use common operations over scalar data, such as linear arithmetic
and Boolean functions.
\sys does not currently support programs that consist of multiple
procedures.
%

\sys for JVM supports program specifications through library
operations that implement the semantics of \cc{assume} and \cc{assert}
instructions. There are further operations for retrieving
non-deterministic data.
The semantics of \lang do not include the concept of a
\cc{NullPointerException} or any other exception. Accessing the fields
of a \cc{null} object is legal, undefined behavior in \lang. \sys does
not model exception-throwing runs of JVM programs. Even so, one can
verify the absence of a \cc{NullPointerException} in the original
program by injecting a \cc{null} check before every field access.

\sys for JVM implements several optimizations, in comparison to the
conceptual version of \sys for \lang given in
\autoref{sec:verifier-algo}.
In particular, the performance of both the procedure \synchc for
synthesizing a run grammar (\autoref{sec:verifier-algo}) and \solvechc
is heavily affected by the number of control locations of a given
program.
\sys for JVM coalesces sequences of value instructions, allocations,
and control branches that are not loop back-edges a single block
to reduce the number of effective control locations.
The benefit of this optimization is that \synskeleton can represent
the space of alternative run-grammar skeletons more compactly, and
\synchc can generate CHC systems with fewer relational predicates and
clauses.
The disadvantage of this optimization is that the constraints in the
generated CHC system are more complex.
In practice, the optimization results in a significant improvement in
performance.

Furthermore, instead of modeling the value of each program variable at
each control point using an uninterpreted function
(\autoref{sec:pf-structures}), \sys translates the program into static
single assignment form, and introduces first-order variables in
generated CHC systems that correspond to program variables.
This increases the total number of first-order variables included in a
system, but allows the \sys implementation to perform fixed-point
analyses such as alias and object liveness analysis that are
relatively cheap compared to the performance cost of a CHC solver.
\sys runs such analyses to simplify the constraints in the CHC clauses
before solving.

Finally, on practical benchmarks, solving a given CHC system to obtain
a counterexample is often far more expensive than an explicit search
through all possible program paths. Accordingly, before attempting to
solve a given system, our implementation of \sys unrolls the system to
enumerate all paths up to some heuristic depth and examines each one
with \autoref{alg:extract} to determine if it is a counterexample.

\sys uses the Sawja program analysis framework~\cite{hubert10} to
process JVM bytecode input and the Z3 theorem prover~\cite{z3} to
solve $\datauif$ satisfiability queries (\autoref{sec:logic}).
\sys uses as a CHC solver (\autoref{sec:verifier-algo}) the
implementation of \duality~\cite{bjorner13} that accompanies Z3,
modified to apply Z3's aggressive formula simplifier after each
iteration of its solving algorithm;
The original version of \duality performs only cheap simplifications,
which often caused a problemtic explosion in the size of the
invariants.
\sys uses the XSB Datalog engine~\cite{xsb} to execute the fixed-point
analyses for optimization prior to solving.
Aside from these components, the \sys codebase consists of
approximately 12,300 lines of OCaml.

We also implemented a \emph{baseline} verifier for JVM bytecode which
encodes programs using the theory of arrays to model the heap.
The baseline generates a CHC system with one relation for each control
location and a clause for each pair of adjacent control locations to
model a step of execution.
The heap of a program is modeled as a collection of logical arrays,
one for each field, that each map each object to the value stored at
that field.
Load and store instructions are modeled with $\mathit{select}$ and
$\mathit{store}$ operations of the theory of arrays, and objects are
represented as integer identifiers.
In this way, the baseline reduces the problem of deciding program
safety to the problem of solving a CHC system in the theory of arrays
with linear arithmetic, \auflia.
The baseline verifier coalesces instructions similarly to \sys as
described above, and attempts to solve the resulting system by running
\duality.

We applied both the baseline and \sys to a set of benchmarks to
determine if \sys could verify the safety of programs that could not
be verified using other theories that accurately model the semantics
of memory operations.
For each benchmark program \cc{P}, we gave \cc{P} to \sys and the
baseline verifier and observed the result of, and resources used by,
each verifier to attempt to prove that \cc{P} satisfies its single
(without loss of generality) \cc{assert} statement.
For most benchmarks, we discovered divergent behavior in the
successive invariants generated by the baseline verifier so that the
solver would inevitably exceed any memory or time limits imposed.
This behavior is reported as failure of the baseline verifier.

We ran all experiments on a machine with 16 2.8 GHz processors and 16
GB of RAM. The current implementations of \sys and the baseline
verifier execute using a single thread.

\subsection{Description of benchmarks}
\label{sec:benchmarks}
The behavior of the thirteen benchmarks programs tested is as follows.

The benchmarks programs \cc{buildInspect}, \cc{peel}, \cc{unary}, and
\cc{binary} perform similar tasks, but with increasing complexity for
\sys. All four programs construct a singly-linked list (using a
\cc{next} field) from front to back, with some property, and then
proceed to scan the list from front to back, ensuring that the
property holds.
In \cc{buildInspect} and \cc{peel}, the property to check is that the
final non-null element scanned is indeed the final element that was
added to the list. In order to model the traversal of this list, \sys
must synchronize the matching iterations of the two loops, as depicted
in \autoref{fig:build-inspect-grammar}. The grammar relations here
each have two negative control pairs, representing control flow of the
two loops. While \cc{buildInspect} has the list initialization peeled
out of the fist loop, in \cc{peel} \sys must discover that this peel
is needed to match the iterations correctly.
A \cc{data} field is added to \cc{unary} and \cc{binary}. The first
checks a unary property over the entire list (that all \cc{data}
fields are zero), while the second checks a binary property (that the
integer \cc{data} is increasing monotonically). A program like
\cc{unary} can be found in the SV-COMP benchmark collection.

The verification of \cc{allocator} constitutes a demonstration that
the encoding of unique object allocation is correct. After an
unbounded list is constructed on the heap, a final object $o$ is
allocated and subsequently compared to every item in the list to check
uniqueness.

\cc{lag2} constructs a list and scans it within the same
loop. However, a given item is scanned exactly two iterations after it
is created.

An unbounded list ending with a fixed-length cycle is built by
\cc{finiteCycle}; it then scans this structure and fails the assertion
if the scan ever terminates.

Alternatively, \cc{breakCycle} constructs a cycle of unbounded size
and proceeds to consume it by destroying the \cc{next} pointers. Thus
it asserts that the traversal terminates when it revisits the node
from which it started. SV-COMP contains a cycle-constructing program
as well. However, the SV-COMP version of the program is fundamentally
easier than \cc{breakCycle} because the data is never overwritten; the
starting point is instead recognized by its integer data.

We adapted \cc{sameLength} directly from a program given by Itzhaky et
al.~\cite{itzhaky13a} which cannot be handled by effectively
propositional reasoning (EPR). The program simultaneously constructs
two disjoint lists, and then simultaneously scans them to check that
they have the same length. This property relies on an invariant
correlating the two disjoint structures, which is not expressible by
heap reachability (see \autoref{sec:ex-benchmarks}).

The \cc{tree} benchmark constructs a non-list data structure. Each
element has two fields, a \cc{left} pointer and a \cc{right}
pointer. A loop non-deterministically allocates a left or right child
for the last item allocated and maintains a counter for the size of
the structure. The program then follows the pointers to the end, and
asserts that the same number of nodes were traversed.

Two more SV-COMP benchmarks which we adapted are \cc{simpleSearch} and
\cc{uniqueItem}. Our \cc{simpleSearch} constructs an unbounded list of
consecutive integers of length at least 5, beginning with zero. It then
expects to find the integers 1 and 3. The original benchmark built a
list of at least ten elements. Our \cc{uniqueItem} constructs an
unbounded list with all \cc{data} fields set to $\false$, except for
exactly one element, which has its \cc{data} set to $\true$. The
program then checks that there is one and only one such element. The
original benchmark performed a more complicated construction, wherein
the unique element is inserted after the list is constructed.

The program named \cc{order} constructs the first part of a list, of
unbounded length, adds two distinguished list elements \cc{a} and
\cc{b}, then constructs the remainder of a list, of unbounded
length. The scanning loop verifies that it does not find \cc{b} before
finding \cc{a}.

Finally, \cc{ctxSensitive} behaves like \cc{lag2} in that the scanning
of elements lags behind their construction in the same loop. However,
two elements are constructed and only one is scanned in each
iteration. Incremental integer data is stored with each element, and
the program asserts that the data value in final element scanned is
half of the data value in the final element constructed.

\subsection{Illustrative benchmark}
\label{sec:ex-benchmarks}
In this section, we discuss in detail an additional benchmark and the operation
of \sys in verifying its safety.
\begin{figure}
  \centering
  \begin{minipage}[t]{.45\textwidth}
    \centering
    \input{code/sameLength.java}
  \end{minipage}
  \qquad
  \begin{minipage}[t]{.45\textwidth}
    \centering
    \input{code/contextSensitive.java}
  \end{minipage}
  \begin{minipage}[t]{.45\textwidth}
    \centering
    \caption{\cc{sameLength}: Constructs two queues in a loop where each
      iteration adds one element to each queue, then traverses both in a
      second loop to show both queues have the same length.}
    \label{fig:eval-sameLength}
  \end{minipage}
  \qquad
  \begin{minipage}[t]{.45\textwidth}
    \centering
    \caption{\cc{ctxSensitive}: A program which \sys cannot prove
      safe.}
    \label{fig:context-sens}
  \end{minipage}
\end{figure}
\paragraph{A program with queues of equal length: \cc{sameLength}}
\autoref{fig:eval-sameLength} contains a program \cc{sameLength} that
builds two queues simultaneously in a loop, then traverses the queues
to check that they contain the same number of elements.
This benchmark was given in a previous study of Effectively
Propositional Reasoning (EPR) in order to illustrate a class of
programs that techniques based on EPR cannot prove safe because their
necessary invariants cannot be represented in Effectively
Propositional Logic (EPL)~\cite{itzhaky13}.
%

\cc{sameLength} performs the following steps over an execution.
\cc{sameLength} first initializes two queues (lines \cc{3}---\cc{7}).
In each iteration of the loop (lines \cc{8}---\cc{16}, referred to as
the \emph{building} loop), \cc{sameLength} adds an element to the tail
of each queue.
\cc{sameLength} non-deterministically chooses to exit the building
loop
It then traverses both queues, loading one element in each iteration of a
loop at lines \cc{17}---\cc{19}, referred to as the \emph{traversal}
loop.
\cc{sameLength} exits the traversal loop when it reaches the end of
one of the queues.
Finally, \cc{sameLength} asserts that it has reached the end of both
queues (line \cc{20}).

In order to prove that \cc{sameLength} satisfies its assertion, a
verifier must establish an invariant for the loop on lines
\cc{17}---\cc{19} that expresses the fact that the lengths of both of
the maintained queues are equal.
Such an invariant cannot be represented in EPL~\cite{itzhaky13}.

\begin{figure*}
  \centering
  \includegraphics[width=.95\linewidth]{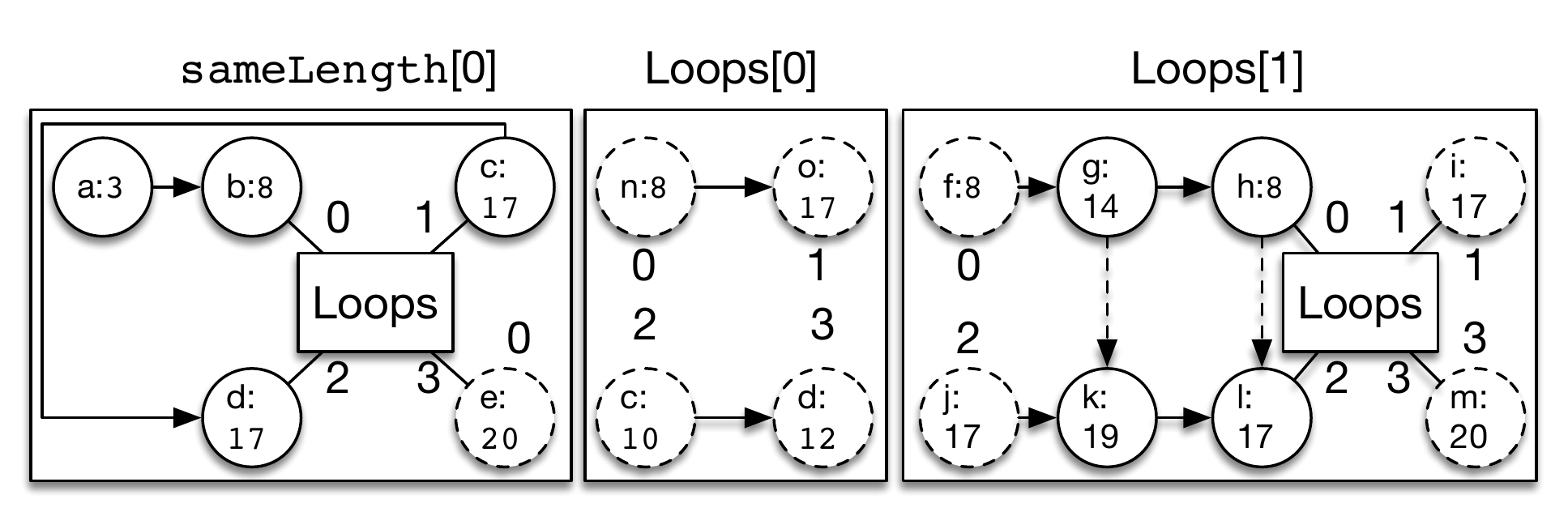}
  \caption{A graph grammar that generates \cc{sameLength}'s control paths
  and data dependencies.}
  \label{fig:sameLength-grammar}
\end{figure*}
\autoref{fig:sameLength-grammar} depicts a run grammar $G$ that
simulates and refutes \cc{sameLength}, and is synthesized
automatically by \sys.
$G$ contains two relational predicates, \cc{sameLength} and
$\mathsf{Loops}$.
The query relational predicate, \cc{sameLength}, is the head of one
clause, $\mathsf{sameLength[0]}$.
The clause derives paths that step from the entry of \cc{sameLength}
to the entry of its building loop (\cc{a}---\cc{b}), %
the step from the exit of the building loop to the entry of the
traversal loop (\cc{c}---\cc{d}), and from the exit of the traversal
loop to the exit of the program (\cc{e}).

The relational predicate $\mathsf{Loops}$ is the head of two clauses.
Clause $\mathsf{Loops}[0]$ derives termination of both the building
and traversal loops.
Clause $\mathsf{Loops}[1]$ derives simultaneous, data-dependent
iterations of the building loop and traversal loop.
The clause derives a control path \cc{f}---\cc{g} in an iteration of
the building loop up to its store that adds an element to the first
queue, %
a control path \cc{g}--\cc{h} from the store that adds to the first
queue to a store that adds to the second queue, %
a control path \cc{j}---\cc{k} that steps in a corresponding iteration
of the traversal loop up to the load from the head of the first queue,
and %
a control path \cc{j}---\cc{l} that steps from the load from the head
of the first queue to a load from the head of the second queue.
\autoref{fig:sameLength-grammar} also depicts data dependencies from
stores at instances of node \cc{g} to corresponding loads at instances
of \cc{k} and from stores at instances of \cc{h}---\cc{l}.
As in previous examples, we omit rules which derive non-corresponding
steps of the loops, to simplify the presentation.

$\mathsf{Loop}$ is empty;
one solution contains an interpretation of $\mathsf{Loop}$ as an
invariant that establishes that the building loop terminates if and
only if both queues are \cc{null} in a corresponding iteration of the
traversal loop:
\begin{align}
  \cc{F}(\cc{head}, \cc{tail}) \equiv
  \cc{tail}_0 = \cc{head}_2 \land
  \cc{next}_1(\cc{tail}_0) = \cc{next}_2(\cc{head}_2) \implies
  \cc{head}_3 = \cc{tail}_1
  \nonumber \\
  (\cc{done}_1 \iff \cc{head1}_3 = \cc{null} \land \cc{head2}_3 =
  \cc{null})
  \land \cc{F}(\cc{head1}, \cc{tail1}) \land
  \cc{F}(\cc{head2}, \cc{tail2})
\end{align}

\cc{sameLength} illustrates the ability of \sys to prove safety of
programs for which safety depends on properties that relate disjoint
unbounded, low-level data structures.
\sys does so by synthesizing invariants over relational predicates
that derive steps of execution that both load and store from related
data structures.
This ability enables \sys to prove the safety of programs that cannot
be proved safe by existing techniques based on
EPL~\cite{drews16,itzhaky13,itzhaky-bj14}.
These approaches can prove properties expressed in terms of heap
reachability within a linked list, but cannot verify properties that
require application-specific inductive
definitions~\cite{itzhaky13}.

\def\ifempty#1{\def\temp{#1} \ifx\temp\empty}
\newcommand{\unitlap}[2]{
  \hphantom{88#1}\mathllap{\ifempty{#2} #2 \else #2\text{#1} \fi}
}
\newcommand{\aligntime}[3]{
  $\unitlap{h}{#1} \, \unitlap{m}{#2} \, \unitlap{s}{#3}$
}

\subsection{Results and analysis}
\label{sec:results}
\settowidth{\rotheadsize}{\synskeleton[\cc{P}]}
\begin{table}[t]
  \begin{tabular}{| l | c | c || c | c | c | c | c | c | c |}
    \hline
    \thead{Benchmark}
    & \thead{Source}
    & \thead{LoC}
    & \thead{Iter}
    & \thead{Rel}
    & \thead{Cls}
    & \thead{\synskeleton\\ Time}
    & \thead{\solvechc \\ Time}
    & \thead{\sys}
    & \thead{Base}
    \\
    \hline
    \hline
    \cc{allocator}    & \multirow{9}{*}{Novel} & 22 & 3 & 7 & 13 &
    \aligntime{}{}{17} & \aligntime{}{4}{30} & \cmark & -- \\
    \cc{binary}       &  & 29 & 7 & 13 & 19 &
    \aligntime{1}{13}{32} & \aligntime{}{5}{25} & \cmark & -- \\
    \cc{buildInspect} &  & 22 & 5 &  7 & 13 &
    \aligntime{}{}{19} & \aligntime{}{1}{18} & \cmark & -- \\
    \cc{ctxSensitive}&    & 21 & -- & -- & -- &
    -- & -- & -- & -- \\
    \cc{finiteCycle}  &  & 19 & 8 & 8 & 15 &
    \aligntime{}{}{59} & \aligntime{}{1}{18} & \cmark & -- \\
    \cc{lag2}         &  & 24 & 4 & 9 & 13 &
    \aligntime{}{}{20} & \aligntime{}{}{6} & \cmark & \cmark \\
    \cc{order}       &    & 32 & 7 & 14 & 24 &
    \aligntime{8}{19}{18} & \aligntime{4}{10}{2} & \cmark & -- \\
    \cc{peel}         &  & 25 & 8 & 10 & 16 &
    \aligntime{}{7}{4} & \aligntime{}{7}{47} & \cmark & -- \\
    \cc{tree}        &  & 38 & 4 & 11 & 16 &
    \aligntime{}{41}{50} & \aligntime{}{16}{18} & \cmark & -- \\
    \hline
    \cc{sameLength}   & EPR   & 28 & 2 & 11 & 17 &
    \aligntime{}{23}{22} & \aligntime{}{7}{54} & \cmark & -- \\
    \hline
    \cc{breakCycle}   & \multirow{4}{*}{SVC} & 25 & 10 & 13 & 21 &
    \aligntime{5}{0}{46} & \aligntime{}{55}{30} & \cmark & -- \\
    \cc{simpleSearch}&  & 33 & 6 & 12 & 18 &
    \aligntime{3}{19}{38} & \aligntime{5}{28}{15} & \cmark & -- \\
    \cc{unary}        &  & 26 & 6 & 12 & 18 &
    \aligntime{}{56}{27} & \aligntime{}{6}{41} & \cmark & -- \\
    \cc{uniqueItem}  &  & 36 & 4 & 11 & 17 &
    \aligntime{}{33}{26} & \aligntime{}{40}{10} & \cmark & -- \\
    \hline
\end{tabular}

\caption{Results of our evaluation of \sys.
  Information about the program source, features of the learned
  grammar, and execution times of \sys are reported. The columns of
  the table indicate the name of the benchmark (``Benchmark''), the
  source of the program (``Source''), the number of lines of source
  code in the program (``LoC''), the number of iterations required to
  learn the final grammar (``Iter''), the final number of relations
  and clauses in the final grammar (``Rel'', ``Cls''), the total time
  spent in calls to the \synskeleton prototype, (``\synskeleton
  Time''), the total time spent in calls to \duality (``\solvechc
  Time''), and whether \sys and the baseline verifiers ultimately
  solved the verification problem (``\sys'', ``Base'').}
\label{tab:eval-results}
\end{table}
The results of our evaluation are contained in
\autoref{tab:eval-results}.
Each of the programs in \autoref{tab:eval-results} was proved safe by
\sys.
Only one of the benchmarks, \cc{lag2}, was proved safe by the baseline
verifier, which took three seconds.

$\synskeleton$ converges to the correct grammar after only 4
infeasible path examples in verifying \cc{buildInspect}, but requires
7 examples for convergence on \cc{peel}.  The $\synskeleton$ time for
\cc{unary} and \cc{binary} is an order of magnitude longer than for
\cc{peel}, which is in turn an order of magnitude longer than for
\cc{buildInspect}. This illustrates the poor scalability of the
constraint-based prototype as the program complexity increases.

Each time an object is allocated in \cc{allocator}, the run grammar
constraints can only ensure that the object is distinct from a bounded
number of other objects. Yet $o$ is verified to be distinct from every
element of the unbounded list, indicating that the earlier allocations
were aware of the eventual occurrence of $o$. This foresight is
achieved because the scanning of the list produces a grammar similar
to that of \cc{buildInspect}. Had the list not been scanned, $o$ could
not be constrained to be distinct from the other elements, which is a
safe approximation precisely because the list is not scanned.

\cc{lag2} stands out as the only one which is determined to be safe by
the baseline verifier, despite the fact that the benchmark uses an
unbounded heap structure. The baseline determines the program is safe
in just three seconds. An informal explanation is that though the
entire structure is unbounded, only a bounded subset of it is
live. The distance between a store and its matching load is bounded at
exactly two iterations of the loop. This special property permits the
invariants to be easily synthesized by \duality in the theory of
arrays as a formula over a bounded number of indices. This same
property also manifests in \sys, without the theory of arrays. The
trivial first $\synchc(\cc{P},\emptyset)$ result is almost
correct: a state variable must be kept in a relation for two
iterations following its use in a control edge.

\cc{finiteCycle} illustrates that \sys can prove cyclicity of the
visited control locations from the cyclicity of a data structure.

Showing correctness of \cc{breakCycle} depends on the ability of \sys
to accurately model a field which is written to and read from more
than once, with an unbounded number of heap modifications in
between. If a verifier could not prove that the nulled pointers remain
null after traversing the unbounded cycle, there would be no guarantee
of ending at the correct element of the cycle.

\sys is well-suited to verify \cc{sameLength}, though it is
inexpressible in EPR, because \sys can capture correlations between
control flow and data movement in the structure of the CHC system.

Verifying \cc{tree} shows that \sys is not confined to lists since the
data structure to be traversed proceeds non-deterministically through
\cc{left} and \cc{right} child pointers.

\sys solved \cc{simpleSearch}, although its CHC system was the hardest
one tested, taking 5.5 hours for \duality to find inductive
invariants. Because of the lower bound on the list size, the shortest
of the feedback paths given to \synskeleton must execute the list
construction step at least 5 times, since the vacuous refuting
neighborhood is all that is required to show infeasibility shorter
paths. Moreover, it is speculated that the CHC solver must also
explore unwindings of the CHC system to that depth before meaningful
invariants begin to be discovered.

\sys is able to solve \cc{uniqueItem}, illustrating its ability to
model a program which can leave a special item at a non-deterministic
point in an unbounded list, and not only find it again, but prove that
there is exactly one such item. That \sys can further solve \cc{order}
shows yet more expressivity in reachability conditions regarding
multiple elements in a list. A verifier would fail if it were only
able to express that the elements \cc{a} and \cc{b} were reachable
from the head of list, and not that the elements are ordered.

\cc{ctxSensitive} is unsolvable by both \sys and the baseline
verifier, despite its similarity to \cc{lag2}, which both verifiers
performed very well on. \sys did not synthesize a correct grammar no
matter how many feedback paths it received. The key difference from
\cc{lag2} is that the distance between a load and its matching store
increases with each iteration. The limitation of \sys in this case is
discussed in below.

\paragraph{Limitations:}
\label{sec:limitations}

While \sys can verify a large class of programs which manipulate heap
data structures, there are certain types of programs which it cannot.
In particular, the current implementation of \sys is restricted to the
subset of programs which can be refuted using a context-free grammar.

\autoref{fig:context-sens} contains the source code for a program,
\cc{ctxSensitive}, for which \sys cannot synthesize a proof.
\cc{ctxSensitive} constructs and consumes a queue in a single loop
where each iteration appends two new elements and consumes only one
previous element. The program checks that the last element read was
the middle element written.
First, the program initializes a queue with a single element and
initializes a counter, \cc{count}, to $1$ (lines \cc{3}---\cc{6}).
The program iterates an arbitrary number of times. In each iteration,
two new elements are added to the tail of the queue. The first element
added has its \cc{data} field set to \cc{count} and the second has its
\cc{data} field set to $\cc{count}+1$ (lines \cc{7}---\cc{13}). Then,
the first element of the queue is removed by moving through the
\cc{head}'s \cc{next} field (line \cc{14}). At the end of each
iteration \cc{count} is incremented by $2$ (line \cc{15}).
Once the loop finishes, the program asserts that the \cc{data} field
of the front of the queue is equal to half of the current value in
\cc{count} (line \cc{16}).

The problem with this benchmark is that the distance between
corresponding stores and loads widens as the program progresses. In
other words, the data dependency is sensitive to the number of
iterations of the loop that have occurred. Therefore, the path-grammar
required to describe the data dependencies of \cc{ctxSensitive} is not
context-free.

Non-context-free grammars are a fundamental limitation of the
presented approach. Extending our technique with the capability of
handling more complex grammars is a compelling direction for future
work.


\section{Related Work}
\label{sec:related-work}
Three-valued logic analysis (TVLA) represents sets of program heaps as
canonical structures~\cite{loginov05,reps04,sagiv02}.
Recent work has introduced automatic shape analyses that represent
sets of program heaps as formulas in separation
logic~\cite{calcagano11,distefano06,reynolds02,yang08}, forest
automata~\cite{holik13}, or memory graphs~\cite{dudka16}.
\sys can be applied to programs that maintain perform arbitrary
low-level heap operations and data operations.
Such approaches can potentially infer invariants of heaps that
maintain a variety of data structures, such as lists or trees.
However, each analyses synthesizes invariants using an abstraction
fixed for the run of the analysis: thus, the analysis can only
potentially be effective if an analysis designer provides a sufficient
set of predicates for maintaining invariants required to prove that a
program satisfies a desired property.
Variants of TVLA have been extended to automatically refine an
abstraction used by applying \emph{Inductive Logic Programming (ILP)}.
However, previous work has established only have ILP can be applied to
refine structural abstractions, not abstractions over structure and
data.

Recent work has also proposed decision procedures for separation
logic~\cite{perez11,seshia03} and an automatic verifier that
represents sets of states as separation-logic
formulas~\cite{albarghouthi15}.
Such approaches require a fixed set of recursively-defined predicates
for reasoning about fixed classes of heap data structures.
\sys can verify program correctness without requiring such predicates.

Previous work has developed automatic verifiers that implement
predicate abstraction in shape logics.
Such approaches can only synthesize invariants that describe which
cells of a heap may reach each other over heap
fields~\cite{balaban05,dams03}, %
can be applied only to programs that maintain particular data
structures---such as linked lists~\cite{lahiri06}---or %
require an analyst to provide loop invariants~\cite{bouillaguet07}, %
predicates over which invariants are
constructed~\cite{flanagan02,rakamaric07,rondon10}, or heuristics that
ensure that the analysis converges~\cite{dams03}.
Relational invariants of path grammars can prove correctness of
programs that maintain non-list data structures.
\sys can potentially synthesize such invariants without requiring an
analyst to provide predicates or loop invariants.

Previous work has proposed verifiers that determine if a program
satisfies an assertion by inferring shape invariants represented as
formulas in effectively-propositional
logics~\cite{drews16,itzhaky13,itzhaky-bj14}.
Such verifiers enjoy strong completeness properties not satisfied by
\sys, but the class of invariants that they synthesize cannot express
invariants over non-list data structures or that relate multiple
lists.
Relational invariants over path grammars can prove the safety of
programs that maintain such data structures, and \sys can potentially
synthesize such invariants automatically.

One such verifier~\cite{itzhaky-bj14} is in fact one instance of a
verification framework that could potentially be instantiated with
logics other than effectively-proposition logics.
However, the general framework requires an analyst to provide
predicates over which verification is to be performed.
\sys does not require such predicates to be provided.

Further work on effectively-propositional logic describes a verifier
that can verify properties of heap-paths of data-structures that are
not necessarily lists~\cite{itzhaky-ba14}.
However, the verifier requires pre-conditions and post-conditions to
be given explicitly.
It cannot verify programs with loops, although loops annotated with
invariants could presumably be verified by directly adapting the
verification technique.
In either case, \sys is distinct from the verifier in that \sys can
potentially prove safety of an iterative program automatically.

Previous work has described verifiers that verify that a given program
satisfies a desired shape property by axiomatizing and inferring
invariants in the theory of arrays.
One approach attempts to infer invariants that are quantified over
array indices and use range predicates that describe the values at
different ranges of indices in an array~\cite{jhala07}.
Such invariants are well-suited to inferring invariants that describe
logical arrays that model arrays operated on by a program, but not
logical arrays that model a program's heap fields.

Previous work has proposed approaches that attempt to verify a given
program by synthesizing a tree-decomposition of the heaps that it
maintains~\cite{manevich07,manevich08}.
\sys is similar to such approaches in that it reasons about
tree-structured artifacts that model program executions, namely the
derivation trees of a grammar of program paths.
Unlike previous approaches, \sys reasons about the derivation of
control paths, rather than of a decomposition of the heap.

Techniques from \emph{relational
  verification}~\cite{barthe11,barthe13,benton04} establish properties
over states in multiple runs of a program, such as robustness and
information-flow security, or over states in runs of multiple
programs, such as observational equivalence.
\sys can be viewed as an instance of relational verification, in that
it attempts to synthesize proofs that establish properties over
multiple states.
\sys is distinct from previous work in that it attempts to establish
properties of states within the same run in order to prove safety of
low-level programs.
Combining existing relational verification techniques with \sys in
order to prove relational properties of programs that maintain
low-level data structures and to automate relational verification
seems to be a promising direction for future work.

\sketch synthesizes finite programs~\cite{solar-lezama06},
bit-streaming programs~\cite{solar-lezama06}, and
stencils~\cite{solar-lezama07} by iteratively synthesizing a candidate
version of a program, attempting to verify it, and using a
counterexample to verification to guide the search for the next
version of the program.
Synthesizers that are instances of the SyGuS framework~\cite{alur13}
attempt to synthesize a program accompanied with a proof of
correctness, using a counterexample-guided, iterative process.
\sys is similar to \sketch and instances of SyGuS in that it is an
inductive synthesizer.
\sys is distinct from \sketch and instances of SyGuS in that each of
them, given a partially-complete program, attempt to complete the
program with instances of a language of possible syntactic
combinations of atomic operations.
\sys, given a program \cc{P}, attempts to synthesize the basic
structure of a logic program that simulates \cc{P}, without
synthesizing new atomic operations.

Previous work has reduced the problem of verifying concurrent programs
that use a bounded number of threads to solving a system of
Constrained Horn Clauses (CHC's)~\cite{gupta11} and proposed a solver
for CHC's over the theory of linear arithmetic.
Previous work has developed automatic verifiers for programs with a
single procedure~\cite{mcmillan06} and multiple recursive
procedures~\cite{heizmann10} that have been generalized to design CHC
solvers that use an interpolating theorem
prover~\cite{bjorner13,rummer13}.
\sys uses a CHC solver for the theory of uninterpreted functions as a
black box.
In principle, \sys can use any CHC solver;
the implementation evaluated uses a solver that itself uses an
interpolating theorem prover~\cite{bjorner13}.


\section{Conclusion}
\label{sec:conclusion}
In this paper, we have presented a novel verifier, named \sys,
designed to verify programs that maintain low-level data structures.
The key feature of \sys is that, given a program \cc{P}, it attempts
to synthesize a proof of the safety of \cc{P}, represented as a graph
grammar, i.e. a CHC system, that generates the control paths of
\cc{P}, annotated with invariants that relate the local values at
multiple points in each path.
\sys is completely automated, using an inductive-synthesis algorithm
that synthesizes candidate path grammars by reduction to constraint
solving and validates candidate grammars by reduction to logic
programming.
Such proofs can establish correctness of programs that previously
could only be proven correct using manually-provided predicates over
values and structure.

We have implemented \sys as a verifier for JVM bytecode and evaluated
it on a set of challenging problems for shape verifiers. \sys for JVM
succeeded in learning a suitable grammar for thirteen of the fourteen
benchmarks. Each of the suitable grammars was given as a CHC system to
the CHC solver \duality, which was able to find inductive invariants
to prove program safety.

The design of \sys establishes that shape-verification problems can
potentially be solved by applying techniques from relational
verification and inductive synthesis.
Further work strengthening this connection could result in significant
development of each of the related fields.

\begin{acks}
  This work is supported by the \grantsponsor{nsf}{National Science
    Foundation}{https://www.nsf.gov} under awards
  \grantnum{nsf}{1526211} and \grantnum{nsf}{1650044}.
\end{acks}

\bibliographystyle{abbrv}
\small
\bibliography{p}

\clearpage

\appendix


\section{Additional Illustrative Benchmarks}
\label{app:more-benchmarks}

\begin{figure}
  \centering
  \begin{minipage}[t]{.5\textwidth}
    \centering
    \input{code/lag.java}
  \end{minipage}
  \qquad
  \begin{minipage}[t]{.4\textwidth}
    \centering
    \input{code/order.java}
  \end{minipage}
  \begin{minipage}[t]{.5\textwidth}
    \centering
    \caption{\cc{lag}: Constructs and consumes a queue, with each
    load occurring one iteration after the corresponding store. }
    \label{fig:eval-lag}
  \end{minipage}
  \qquad
  \begin{minipage}[t]{.4\textwidth}
    \centering
    \caption{\cc{order}: Constructs and queue with two distinguished
      elements. It then consumes the queue to ensure the distinguished
      elements remain in the order they were inserted. }
    \label{fig:eval-order}
  \end{minipage}
\end{figure}

\begin{figure*}[t]
  \centering
  \includegraphics[width=0.9\linewidth]{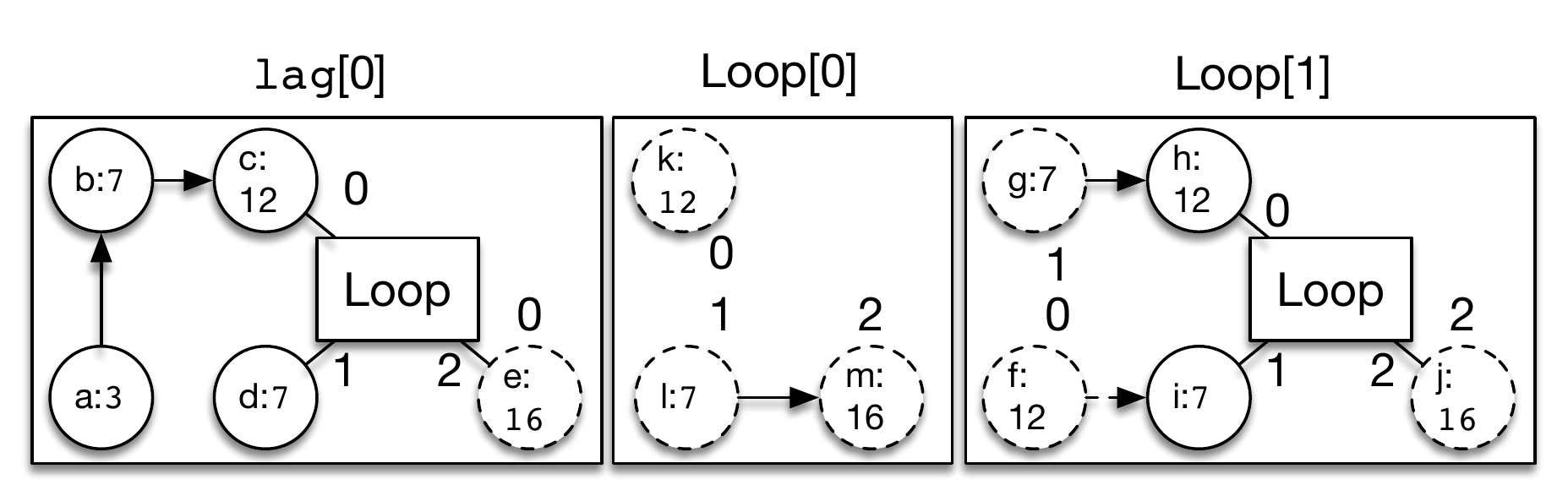}
  \caption{A graph grammar that generates \cc{lag}'s control paths and
  data dependencies.}
  \label{fig:lag-grammar}
  \BH{Figure is missing two control edges: One from c->d and one from
      h->i, both of which represent stepping back to the beginning of the
      loop.}
\end{figure*}
\paragraph{A program with nodes shared across clauses: \cc{lag}}
\autoref{fig:eval-lag} contains the source code for a program,
\cc{lag}, which builds and consumes a queue. Both the loads and stores
occur in a single loop, where each item that is loaded was stored on
the previous iteration. The program operates in the context of a small
state machine.
First, \cc{lag} initializes the queue and sets itself in a
\emph{write-only}
state (lines \cc{4}---\cc{6}).
Next, if \cc{lag} is in a \emph{write} state, then it adds a new element to
the end of the list and sets the tail to this new element (lines
\cc{7}---\cc{11}).
If \cc{lag} is in a \emph{read} state, it moves the head to the next element
(line \cc{12}).
\cc{lag} then performs state transitions (lines \cc{13}---\cc{15}).
\cc{lag} moves from the \emph{write-only} state to the
\emph{read-write} state and from the \emph{read-only} state to the
\emph{done} state after a single iteration. \cc{lag} moves from the
\emph{read-write} state to the \emph{read-only} state after an
arbitrary number of iterations. Finally, \cc{lag} asserts that the
head element reached by this loop is the same as the tail (line
\cc{16}).

\autoref{fig:lag-grammar} depicts a graph grammar, $\mathcal{G}$,
which describes the control paths and data dependencies of
\cc{lag}. The grammar contains two relations, \cc{lag} and \cc{Loop},
where \cc{lag} is the query relation. $\mathcal{G}$ has three clauses.
$\mathsf{lag}[0]$ generates fragments of the control path which
describe the loop entry (\cc{a}---\cc{b}), through the first store
(\cc{b}---\cc{c}), and through the first state transition without
performing a load (\cc{c}---\cc{d}).
$\mathsf{Loop}[0]$ generates fragments of the control path which store
the next item in the queue (\cc{g}---\cc{h}), load the front of the
queue, and perform a state transition (\cc{h}---\cc{i}).  In addition,
$\mathsf{Loop}[0]$ contains a data dependency (\cc{f}---\cc{i}): The
store from the previous loop iteration provides the data for the load
in this iteration.
$\mathsf{Loop}[1]$ generates the final fragment which skips from the
beginning of the loop to just after the loop (\cc{l}---\cc{m}).

A noteworthy feature of $\mathcal{G}$ is node \cc{f}. In particular,
\cc{f} provides a data dependency, but is not connected directly to
the control path through $\mathsf{Loop}[0]$. This demonstrates the
capability of graph grammars to separate data dependencies from
control dependencies.
%

\begin{figure*}[t]
  \centering
  \includegraphics[width=.95\linewidth]{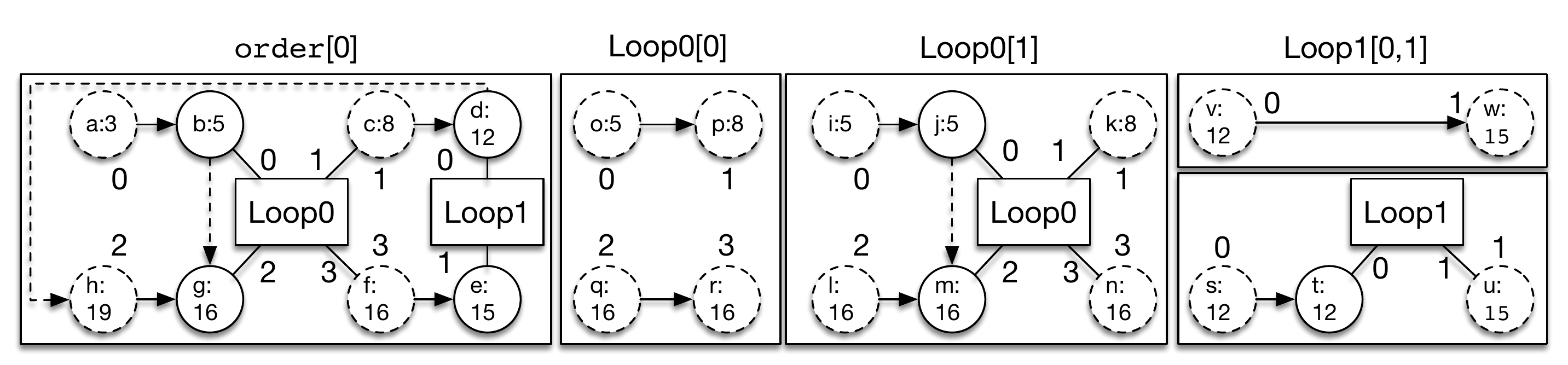}
  \caption{A graph grammar that generates \cc{order}'s control paths
  and data dependencies.}
  \label{fig:order-grammar}
  \BH{control flow between h/g, f/e is backwards. Should be g->h, e->f}
\end{figure*}
\paragraph{A program for which some heap operations have no data dependencies:
\cc{order}}
\autoref{fig:eval-order} contains the source code for a program,
\cc{order}, which creates a queue with two distinguished elements,
then consumes the queue until it can check these two elements remain
relatively ordered.
First, \cc{order} initializes the queue and adds an arbitrary number
of elements to the tail (lines \cc{3}---\cc{7}). Next, it inserts two
distinguished elements, \cc{a} and \cc{b} (lines \cc{8}---\cc{11}).
\cc{order} continues to expand the queue with an arbitrary number of
new elements (lines \cc{12}---\cc{14}). \cc{order} traverses the
queue until it reaches element \cc{a}. If \cc{b} is encountered before
the end of this loop, then the queue has become unordered (lines
\cc{15}---\cc{18}). Finally, \cc{order} asserts that the distinguished
elements remain in order (line \cc{19}).

Each execution of \cc{order} satisfies the assertion. The key
invariant which could prove the safety property establishes that for
each iteration of the loop on lines \cc{12}---\cc{14}, traversing the
\cc{next} pointer reaches $\cc a$ before it reaches $\cc b$. This
invariant is difficult to express without knowing particular facts
about the relationships of linked lists. However, a graph grammar can
express the idea of this invariant by simulating the first and third
program loops together.

\autoref{fig:order-grammar} depicts a graph grammar, $\mathcal{G}$,
which encodes \cc{order}. $\mathcal{G}$ has three relations and five
clauses. The query relation, \cc{order}, has only one clause,
$\cc{order}[0]$, which contains the entry to the first program loop
(\cc{a}---\cc{b}), the insertion of the two distinguished elements and
entry of the second loop (\cc{c}---\cc{d}), the entry to the third
loop (\cc{e}---\cc{f}), a final execution of the third loop and the
program assertion (\cc{g}---\cc{h}), and a data link between the first
distinguished element and the final load (\cc{d}---\cc{h}).
The two clauses of \cc{Loop0}, $\mathsf{Loop0}[0]$ and
$\mathsf{Loop0}[1]$, describe the simultaneous execution of the first
and third program loops. These clauses are similar to
$\mathsf{StLd}[0]$ and $\mathsf{StLd}[1]$ from
\autoref{sec:ex-grammar}.
The two clauses of \cc{Loop1}, $\mathsf{Loop1}[0]$ and
$\mathsf{Loop1}[1]$ describe the execution of the second program loop.
For brevity, we omit some clauses for \cc{Loop0} that iterate the two
loops a differing number of times. We annotate \cc{Loop0} with an
inductive relational invariant which can prove the safety of
\cc{order}:
\begin{align}
  \cc{tail}_1 \neq \cc{a}_2 \land \cc{head}_3 = \cc{tail}_1 \land
  \cc{ordered}_3 \land
  \cc{tail}_0 = \cc{head}_2 \land \nonumber \\
  \cc{next}_1(\cc{tail}_0) = \cc{next}_2(\cc{head}_2) \implies
  \cc{head}_3 = \cc{tail}_1
\end{align}
This invariant maintains that the flag \cc{ordered} is always true at
the end of third loop. In addition, it establishes that the
simultaneous iteration of the first and third loop terminates before
the distinguished element $a$ is reached. Interestingly, the second
loop is irrelevant to proving the safety property, so the relational
invariant for \cc{Loop1} is simply $true$.

Previous examples have all contained data dependencies for every store
and load. \cc{order} demonstrates that this is not a limitation of
\sys. The stores performed in the second program loop have no
corresponding loads. This is reflected by the lack of data
dependencies in the clauses of the relation \cc{Loop1}.

\paragraph{A program with a tree data structure: \cc{tree}}

\begin{figure}[t]
  \centering
  \input{code/tree.java}
  \caption{\cc{tree}: Constructs a tree by arbitrarily adding left or
    right nodes, then follows the path to ensure the length when
    consumed is the same as the length of the constructed path.}
  \label{fig:eval-tree}
\end{figure}

\autoref{fig:eval-tree} contains a program, named \cc{tree}, that
constructs a tree data structure.
\cc{tree} builds the tree by non-deterministically choosing to
continue to build the \cc{left} or \cc{right} subtree of a given node
for a non-deterministically chosen number of iterations.
\cc{tree} then traverses the data structure and ensures that the
number of steps taken to traverse it is the same as the number of
steps used to build it.

In particular, \cc{tree} enters the loop in lines \cc{7}---\cc{14}
(i.e., it's \emph{building} loop) by initializing the tree and
initializing a counter \cc c of the number of nodes.
Before each iteration of the loop, \cc{tree} tests if the tree node
bound to \cc{t} is \cc{null}, and if so, exits the loop (line \cc{7}).
In each iteration of the loop, \cc{tree} increments \cc c
(line \cc{8}) and non-deterministically chooses whether to stop
adding new nodes (lines \cc{9}---\cc{10}).
The new node is then stored in either the \cc{left} or \cc{right}
field of the node bound to \cc{t} (lines \cc{11}---\cc{13}).

After executing the loop, \cc{tree} binds \cc{t} to the root of the
build tree (line \cc{15}) and initializes a new counter \cc{d} to $0$
(line \cc{16}).
When \cc{tree} executes the loop on lines \cc{17}---\cc{21} (i.e., its
\emph{traversal} loop), it traverses the constructed path by loading
either the \cc{left} or \cc{right} child of a node maintained in
\cc{t}, depending on which child is not \cc{null} (lines
\cc{17}---\cc{20}).
\cc{tree} executes its traversal loop until it reaches a node such
that both children are \cc{null} (line \cc{21}).
\cc{tree} then asserts that the number of nodes added is the same as
the number of nodes traversed (line \cc{22}).

\cc{tree} satisfies the safety property tested by its assertion at
line \cc{22}
However, the loop invariants required to prove the property must
express a complex property that relates data structure and values.
Such invariants must establish that the tree contains \cc{c} elements
while \cc{tree} executes the loops in lines \cc{7}---\cc{14}, and that
$\cc{c} - \cc{d}$ elements are reachable from the object in \cc{t}
while \cc{tree} executes the loop in lines \cc{17}---\cc{21}.
To our knowledge, no automated invariants can synthesize proofs in a
language that can express such properties.
%

\begin{figure*}[t]
  \centering
  \includegraphics[width=.95\linewidth]{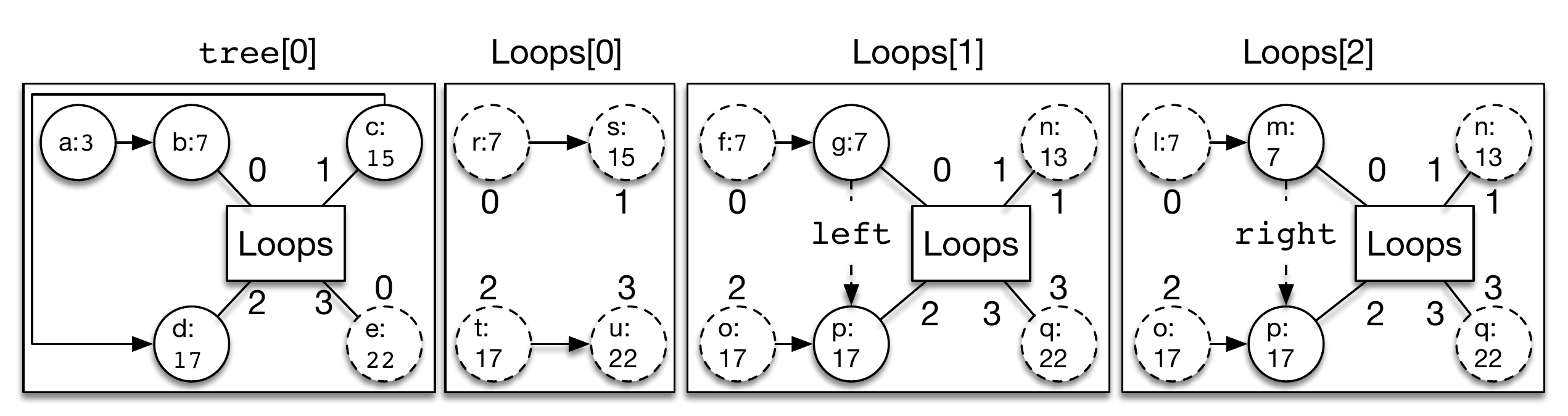}
  \caption{A graphical depiction of a run grammar that simulates and
    refutes \cc{tree}.
    Control and data dependencies between path points are depicted
    similarly to \autoref{fig:build-inspect-grammar}.}
  \label{fig:tree-grammar}
\end{figure*}
\autoref{fig:tree-grammar} depicts a run grammar $\mathcal{G}$ that
simulates and refutes \cc{tree}.
$\mathcal{G}$ contains two relational predicates, \cc{tree} and
$\mathsf{Loops}$.
The query relational predicate, \cc{tree}, is the head of clause
$\cc{tree}[0]$.
Clause $\cc{tree}[0]$ derives a control path that steps from program
entry to the entry of \cc{tree}'s building loop (\cc{a}---\cc{b}),
a path that steps from the exit of the building loop to the entry of
the traversal loop (\cc{c}---\cc{d}), and a path that steps from the
exit of the traversal loop to the assertion (\cc{e}).

$\mathsf{Loops}$ is the head of three different clauses.
Clause $\mathsf{Loops[0]}$ derives steps in which the building loop
non-deterministically chooses to exit;
in the corresponding step of traversal loop, it determines that it has
reached the end of the tree.
Clause $\mathsf{Loops[1]}$ derives simultaneous steps of both loops
such that the building loop chooses to build a subtree from the
\cc{left} field of the maintained tree node.
\cc{f}---\cc{g} derives a path through the building loop, while
\cc{i}---\cc{j} depicts the corresponding path through the traversal
loop.
Each path point derived as an instance of \cc{g} is a data dependence
of the node derived as an instance of \cc{j} in the same clause
instance, depicted by the data-dependence edge from \cc{g} to \cc{j}.
Clause $\mathsf{Loops}[2]$ is similar to $\mathsf{Loops}[0]$, but
derives corresponding steps of both loops in which the building loop
builds a tree at the \cc{right} field of the maintained tree node.
Similar to the presentation of \cc{buildInspect} in
\autoref{sec:overview}, we have omitted clauses that derive extraneous
iterations of either the building or traversal loops.
Such clauses are necessary to represent a run grammar that simulates
all control flow paths of \cc{tree}, since it is not known \emph{a
  priori} that every execution has an equal number of iteration in
each loop.

$\mathcal{G}$ simulates \cc{tree} and is empty;
one solution of $\mathcal{G}$ includes an interpretation of
$\mathsf{Loops}$ as a relational invariant that establishes that
\cc{c} and \cc{d} have the equal values at the end of iterations on
which they store to and load from the same field:
\begin{align}
  \cc{lrnull}(ix0, ix1) \equiv
  \cc{left}_{ix0}(\cc{t}_{ix1}) = \cc{null} \land
  \cc{right}_{ix0}(\cc{t}_{ix1}) = \cc{null}
  \nonumber \\
  \cc{c}_1 = \cc{d}_3 \land (\cc{lrnull}(1, 0) \iff \cc{lrnull}(2, 2))
\end{align}
In this invariant we define a formula macro \cc{lrnull} over two
indices which specifies that the \cc{left} and \cc{right} fields of
the variable \cc{t} are both null at the given indices.
To our knowledge, no automatic shape verifier developed in previous
work can prove the safety of \cc{tree} without being guided to use a
manually-defined recursive predicate that describes the shape of
trees, and relevant predicates over data variables \cc{c} and \cc{d}.
However, \sys can prove the safety of \cc{tree} automatically.

\section{Implementing \synskeleton by reduction to constraint solving}
\label{app:syn-skeleton-ctrs}

We now formalize some concepts that will assist in proving the lemmas
and explaining more precisely how the constraints used to implement
\synskeleton work.
In the following sections, let a program $\cc{P} \in \lang$ be
given. Let $\mathcal{C}$ be a CHC system, and let the space to which
$\mathcal{C}$ belongs be denoted $\chcs{\mathcal{R}}{X}$. Assume,
without loss of generality, that the arity of every $R \in
\mathcal{R}$ is $n$. Let $\idxs = \setformer{i \in \nats}{i \le
  n}$. Let $\alpha(Y,i)$ be the $i$th item in the sequence $Y$ for $Y
\in X^n$ and $i \in \idxs$. Further, for every control location $l \in
\locs[\cc{P}]$, let $\cc{L}_l$ be a unary interpreted
function. Further, let $\cc{Ctrl}$, $\cc{Data}$, and $\cc{Conn}$ be
binary uninterpreted functions.
\paragraph{Location consistency}
Suppose there exists $\hat\lambda : \mathcal{R} \times \idxs \to
\locs$ with the following property for all clauses $C = (B, \varphi,
H) \in \mathcal{C}$, all applications $(R,Y) \in B \union \{H\}$, and
all indices $i \in \idxs$:
\begin{align*}
  \varphi \entails \bigland_{l \in \locs} {
    (l=\hat\lambda(R,i) \iff \cc{L}_{l}(\alpha(Y,i)))
  }.
\end{align*}
Then $\mathcal{C}$ is \emph{location consistent}.
Also define the map ${\hat\lambda}_C : X \to \locs$ for each clause $C
= (B, \varphi, H) \in \mathcal{C}$ such that if $\varphi \entails
\cc{L}_l(x)$ then ${\hat\lambda}_C(x) = l$, for all variables $x \in
X$ and locations $l \in \locs$.
${\hat\lambda}_C$ is well-defined because $\mathcal{C}$ is location
consistent.

\paragraph{Control consistency}
Suppose $\mathcal{C}$ is location consistent.
Let $\hat\lambda_C$ be given for all clauses $C \in \mathcal{C}$.
Suppose there exists ${\hat E}_C \subseteq X \times X$ for each clause
$C = (B, \varphi, H) \in \mathcal{C}$ with the following properties
for all variables $x_1,x_2 \in X$:
\begin{enumerate}
\item 
  If there does not exist an instruction
  $\instrat{\cc{P}}{{\hat\lambda}_C(x_1)}{{\hat\lambda}_C(x_2)}$, then
  $\varphi \entails \lnot \cc{Ctrl}(x_1, x_2)$;
\item
  $\varphi \entails {\hat E}_C(x_1,x_2) \iff \cc{Ctrl}(x_1, x_2)$.
\end{enumerate}
Then $\mathcal{C}$ is \emph{control consistent}.

\paragraph{Path consistency}
Suppose $\mathcal{C}$ is control consistent, and let ${\hat E}_C$ be
given for all clauses $C \in \mathcal{C}$ Suppose there exist a total
ordering $\ord{R}$ over $\idxs$ and a map $\sgn{R} : \idxs \to
\{+,-\}$ for every relation $R \in \mathcal{R}$, such that the
following properties hold for all clauses $C = (B, \varphi, H) \in
\mathcal{C}$, all applications $(R_1,Y_1), (R_2, Y_2), \in B \union
\{H\}$, all indices $i,j,k \in \idxs$, and all variables $x_1,x_2,x_3
\in X$:
\begin{enumerate}
\item If $(\alpha(Y_1,i),\alpha(Y_2,j)) \in {\hat E}_C$ then $\varphi
  \entails \cc{Conn}(\alpha(Y_1,i),\alpha(Y_2,j))$ and
  \begin{enumerate}
  \item if also $(R_1,Y_1) = H$ then
    $\varphi \entails \sgn{R_1}(i) < 0$,
  \item if also $(R_1,Y_1) \in B$ then
    $\varphi \entails \sgn{R_1}(i) > 0$;
  \end{enumerate}
\item If $(\alpha(Y_1,i),\alpha(Y_2,j)) \not\in {\hat E}_C$ and if
  $\varphi \entails \alpha(Y_1,i) = \alpha(Y_2,j)$ then $\sgn{R_1}(i)
  = \sgn{R_2}(j)$;
\item $\varphi \entails {
  {\cc{Conn}(x_1,x_2) \land
    \cc{Conn}(x_2,x_3)} \implies
  \cc{Conn}(x_1,x_3)
}$;
\item $\varphi \entails i \ord{R_1} j \iff
  \cc{Conn}(\alpha(Y_1,i),\alpha(Y_1,j))$.
\end{enumerate}
Then $\mathcal{C}$ is \emph{path consistent}.

For a path consistent $\mathcal{C}$, define sets of index pairs
$\posfrags{R}$ and $\negfrags{R}$ for each relation $R \in
\mathcal{R}$ as follows. Let $\posfrags{R} =
\setformer{(i,j)}{\sgn{R}(i) > 0 \land i \ord{R} j \land \not\exists k
  \in \idxs (i \ord{R} k \ord{R} j)} \subseteq \idxs \times
\idxs$. Define $\negfrags{R}$ similarly, but with $\sgn{R}(i) <
0$. The sets $\posfrags{R}$ and $\negfrags{R}$ are the \emph{positive}
and \emph{negative control pairs} of $R$, respectively. Informally, at
any node $d$ of any derivation tree $D$ of $\mathcal{C}$, with head
$(R,Y)$, the control path between $\alpha(Y,i)$ and $\alpha(Y,j)$ is
constructed entirely within the subtree of $D$ rooted at $d$ if $i
\negfrags{R} j$. Similarly, the control path is constructed entirely
outside of the subtree of $D$ rooted at $d$ if $i \posfrags{R} j$.
For convenience, define the application of $\posfrags{R}$ to variable
tuple $Y \in X^n$ as the set of variable pairs $\posfrags{R}(Y) =
\setformer{(\alpha(Y,i),\alpha(Y,j))}{i \posfrags{R} j,\; i,j \in
  \idxs} \subseteq X \times X$. Define $\negfrags{R}(Y)$ similarly.

\paragraph{Neighborliness}
Let $\mathcal{C}$ be such that for each $C \in \mathcal{C}$ and all
$x_1,x_2 \in X$,
\begin{align*}
  \varphi \entails \cc{Data}(x_1, x_2).
\end{align*}
Then $\mathcal{C}$ is \emph{neighborly}.

\paragraph{Completeness and Correctness}

If every path $p \in \pathsof{\cc{P}}$ is induced as the control graph
of some model of $\mathcal{C}$, then $\mathcal{C}$ is \emph{complete}
for \cc{P}.

Let $\cc{p} \in \pathsof{ \cc{P} }$ be infeasible.
If each derivation of $\mathcal{C}$ that induces $p$ as a control
graph also induces a $\nu$ a refuting neighborhood of $p$, then
$\mathcal{C}$ is \emph{correct} for $p$.
If $\mathcal{C}$ is complete for \cc{P} and correct for all infeasible
paths $p \in \pathsof{\cc{P}}$, then $\mathcal{C}$ is \emph{correct}
for \cc{P}.

There is a direct correspondence between the models $m \sats \sympath
(\cc{P}, p, \nuAll)$ and runs of $p$.
Model $m$ corresponds to run $(p, \sigma)$ with $\sigma$ defined by
$\sigma(n) = (\gamma^n, U^n)$, where $\gamma^n(\cc{a}) =
m(\cc{a}(\sigma_n))$ and $U^n(\cc{f},\cc{x}) = n'$ where
$m(\cc{f}(\tnow(\sigma_n), \cc{x}(\sigma_n))) = m(\tnow(\sigma_{n'}))$
if any such $n'$ exists, for all $n \in N, \cc{a} \in \vars, \cc{x}
\in \objvars, \cc{f} \in \fields$.

\paragraph{Generating a skeleton from constraint solutions}

The prototype implementation of \synskeleton is a constraint-based
approach to grammar synthesis. We encode a space $\mathbf{G}$ of
grammars using uninterpreted functions to model the structure of
relations and clauses.

With loss of generality, the prototype only considers linear grammars,
in the sense that every clause has at most one relation in the
body. \sys has been formulated for only single procedure programs,
whose control flow can be described linearly, and a need for
non-linear grammars has not yet been encountered in \sys. The
prototype also requires as input a budget specifying the maximum arity
of relations and the maximum number of relations for systems in
$\mathbf{G}$. This configuration is encoded as the set of constants
and uninterpreted functions of an SMT query.

The following properties of every $\mathcal{C} \in \mathbf{G}$ are
achieved by SMT constraints:
\begin{enumerate}
\item Location consistency: This gives rise to ${\hat\lambda}_C : X
  \to \locs$ for each clause $C$.
  
\item Control consistency: WLOG, each clause $C$ entails exactly one
  control edge, i.e., ${\hat E}_C = \{(a,b)\}$.

\item Path consistency: The arguments of every relation are
  partitioned into some number of negative control pairs and some
  number of single \emph{auxiliary} variables.
  Let $(R_2, \varphi, R_1)$ be a clause.
  Under path consistency, $R_1$ is responsible for completing the
  control paths between its negative pairs, and $R_2$ inherits what is
  not explicitly completed by $R_1$. Assume that $R_1$ constructs the
  control edge $(a,b)$ between negative pair $(a',b')$. If
  $(a,b)=(a',b')$, then this pair is completed and $a,b$ become
  auxiliary variables in $R_2$. If $b \ne b'$, then only $a'$ becomes
  auxiliary and $a \negfrags{R_2} b'$. Likewise for $a \ne a'$. If
  both $a \ne a'$ and $b \ne b'$, then the control pair is split in
  two in $R_2$: $a' \negfrags{R_2} a \posfrags{R_2} b \negfrags{R_2}
  b'$.
  Auxiliary variables (those not in any negative pair) may or may not
  be forwarded to $R_2$. In the general case, for $x \posfrags{R_1} y
  \posfrags{R_1} z$, when $y$ is not forwarded, $x \posfrags{R_2} z$.
\item Completeness as a path grammar:
  For each clause with control edge $(a,b)$, ${\hat\lambda}_C(a)$ is
  uniquely determined by the head relation.  For each relation $R$,
  there is (at least) one clause with head $R$ for each possible value
  of ${\hat\lambda}_C(b)$. The query relation contains the unique
  (WLOG) initial and final control locations as a negative control
  pair. By induction over the clauses, each $R$ completes all possible
  paths between every negative control pair.
\item Unambiguity (with loss of generality): For each relation $R$,
  there is at most one clause with head $R$ for each possible value of
  ${\hat\lambda}_C(b)$, guaranteeing that every control path $p$ has a
  unique derivation.
\item Correctness for a path $p \in \feedback$: A collection of
  uninterpreted functions witnesses a derivation of $p$ by the grammar
  and further witnesses that the neighborhood $\nu =
  \extractdeps(\cc{P},p)$ is achieved, assuming neighborliness. That
  this holds for all derivations of $p$ is trivial with unambiguity.
\end{enumerate}

\section{Proofs of Lemmas and Theorems}
\setcounter{lemma}{0}
\setcounter{thm}{0}

\begin{lemma}
  \label{lemma:pf-corr:restate}
  If \cc{P} is simulated by $\mathcal{S}$ and \cc{P} is refuted by
  $\mathcal{S}$, then \cc{P} is safe (\autoref{sec:semantics}).
\end{lemma}
\begin{proof}
  $\mathcal{S}$ has no model that is a run \cc{P}, by the assumption
  that \cc{P} is refuted by $\mathcal{S}$.
  Therefore \cc{P} has no run by the assumption that \cc{P} is
  simulated by $\mathcal{S}$.
  Thus \cc{P} is safe, by the definition of safety
  (\autoref{sec:semantics}).
\end{proof}

\begin{lemma}
  \label{lemma:is-feasible-corr:restate}
  If \cc{p} is a feasible path of \cc{P}, then $\isfeasible(\cc{P},
  \cc{p}, \nuAll) = \true$.
  Otherwise, $\isfeasible(\cc{P}, \cc{p}, \nuAll) = \false$.
\end{lemma}
\begin{proof}
  \textit{Sketch.} \isfeasible uses the formula $\sympath(\cc{P},
  \cc{p}, \nuAll)$, which is a conjunction of constraints which
  precisely model the semantics of each instruction in $p$. Notably,
  for every load there is at most one matching store. Every load is
  constrained with the value of its matching store since the
  neighborhood $\nuAll$ allows for every control state of $p$ to be
  examined in search of the unique match. Thus, the semantics of \lang
  are not approximated, and $\sympath(\cc{P}, \cc{p}, \nuAll)$ is
  satisfiable exactly when $p$ has a run.
\end{proof}

\setcounter{lemma}{3}
\begin{lemma}
  \label{lemma:sympath-weak:restate}
  Let $p = (N,\lambda,E) \in \pathsof{\cc{P}}$ and a neighborhood
  $\nu_0 : N \to \pset(N)$ be given. $\sympath(\cc{P},p,\nu_0)
  \entails \sympath(\cc{P},p,\nuAll)$.
\end{lemma}
\begin{proof}
  $\sympath(\cc{P},p,\nu)$ is a conjunction of constraints modeling
  the instruction on each edge $(u,v) \in E$, which is
  $\symrel{i}(u,v,\nu(v))$ for the instruction $i =
  \instrat{\cc{P}}{\lambda(u)}{\lambda(v)}$. For every $n \in N$,
  $\nu_0(n) \subseteq \nuAll(N) = N$. We have that
  $\symrel{i}(u,v,\nu_0(v)) \entails \symrel{i}(u,v,N)$, shown
  case-wise for every possible instruction $i$:
  \begin{itemize}
  \item If $i \in \valinstrs$, then $\symrel{i}(u,v,Q)$ does not
    depend on the set of states $Q$.
  \item If $i \in \stores \union \loads \union \allocs$, then
    $\symrel{i}(u,v,Q)$ takes only a conjunction over the $q \in
    Q$. Since $\nu_0(v) \subseteq N$, $\symrel{i}(u,v,\nu_0(v))$ is no
    stronger than $\symrel{i}(u,v,N)$.
  \end{itemize}
  Since $\sympath$ is only a conjunction of $\mathsf{SymRel}$
  constraints, $\sympath(\cc{P},p,\nu_0)$ cannot be any stronger than
  $\sympath(\cc{P},p,\nuAll)$.
\end{proof}

\begin{lemma}
  \label{lemma:rungram-corr:restate}
  Let a control path $p = (N,\lambda,E) \in \pathsof{\cc{P}}$ and any
  neighborhood graph $(N, \dataedges)$ be given. Let $\nu : N \to
  \pset(N)$ be the neighborhood induced by $(N,\dataedges)$. Suppose
  $p$ with $\nu$ are modeled by some derivation $D$ of a skeleton
  $\mathcal{C}$. Let $\mathcal{G}$ be the run grammar corresponding to
  $\mathcal{C}$. Then there is a model of $D$ in $\mathcal{G}$ iff
  there is a model of $\sympath(\cc{P},p,\nu)$.
\end{lemma}
\begin{proof}
  Let a model of $D$ in $\mathcal{G}$ be given as $(D,m_D,i)$, where
  $m_D$ is a model of the background theory and $i : \nodesof{D} \to
  \modelsof{X}$ maps derivation nodes to models of the CHC variables
  $X$. Since $\mathcal{C}$, a skeleton, is control consistent by the
  specification of \synskeleton, every control edge $(u,v) \in E =
  m_D(\cc{Ctrl})$ via the constraint $\varphi$ of the clause $C$ of
  some node $d$ in $D$, as follows: Let ${\hat\lambda}_C, {\hat E}_C$
  be given. There exists a pair $(x_1,x_2) \in {\hat E}_C$ such that
  $i(d)(x_1) = u, i(d)(x_2) = v$ and $\varphi \entails
  \cc{Ctrl}(x_1,x_2)$. There is a corresponding clause in
  $\mathcal{G}$ with a constraint $\psi \entails \varphi$. Moreover,
  because $\psi \entails
  \symrel{\instrat{\cc{P}}{{\hat\lambda}_C(x_1)}{{\hat\lambda}_C(x_2)}}(x_1,x_2,X)$,
  and because $\mathcal{G}$ is location consistent and neighborly,
  $i(d) \sats
  \symrel{\instrat{\cc{P}}{\lambda(u)}{\lambda(v)}}(u,v,\nu(v))$. Let
  $m$ be a model over language of $\sympath(\cc{P},p,\nu)$ such that
  $m$ restricted to $\{u,v\} \union \nu(v)$ is $i(d)$. Since this
  property of $m$ holds for all control edges $(u,v) \in E$ then $m
  \sats \sympath(\cc{P},p,\nu)$ by the definition of $\sympath$.

  Now, let a model $m$ of $\sympath(\cc{P},p,\nu)$ be given. The
  process above can almost be reversed to obtain models $m_D$ and
  $i(d)$ for each $d \in \nodesof{D}$, but for each clause in
  $\mathcal{G}$ with constraint $\psi$, care must be taken surrounding
  the extra condition which $\psi$ entails, described in
  \autoref{sec:footprint}. Because the encoding of loads and stores
  requires that every live object be held by some local state
  variable, if an object is not passed from $R$, resp. to $R$, this
  object is guaranteed to not have any of its fields modified on the
  control paths between any of the $\posfrags{R}$ pairs,
  resp. $\negfrags{R}$ pairs. Therefore, $\psi$ does not overconstrain
  $i(d)$, and we have a model $(D,m_D,i)$ of $\mathcal{G}$ inducing
  $p$ and $(N,\dataedges)$.
\end{proof}

\setcounter{lemma}{2}
\begin{lemma}
  \label{lemma:synchc-corr:restate}
   \cc{P} is simulated by \synchc(\cc{P}, \feedback).
\end{lemma}
\begin{proof}
  \synchc begins by creating a CHC system $\mathcal{C} =
  \synskeleton(\cc{P}, \feedback)$. By the specification of
  \synskeleton, $\mathcal{C}$ is complete for \cc{P}; i.e., every path
  $p \in \pathsof{\cc{P}}$ is the control graph modeled by some
  derivation of $\mathcal{C}$.

  \synchc then constructs the run grammar $\mathcal{G}$ corresponding
  to $\mathcal{C}$. By \autoref{lemma:rungram-corr:restate}, for every
  path $p \in \pathsof{\cc{P}}$, there is a derivation $D$ of
  $\mathcal{G}$ iff there is a model of $\sympath(\cc{P},p,\nu)$
  for some neighborhood $\nu$. By \autoref{lemma:sympath-weak:restate}
  and \autoref{lemma:is-feasible-corr:restate}, there is a model of
  $D$ iff there is a run of $p$.
  Therefore, $\mathcal{G}$ simulates every path $p \in
  \pathsof{\cc{P}}$, and thus simulates \cc{P}.
\end{proof}

\begin{thm}
  \label{thm:soundness:restate}
  If $\sys(\cc{P}) = \true$, then \cc{P} is safe, and if $\sys(\cc{P})
  = \false$, then \cc{P} is not safe.
\end{thm}
\begin{proof}
  \sys returns $\true$ only if \solvechc finds a solution to
  $\mathcal{G} = \synchc(\cc{P},\feedback)$, which proves that
  $\mathcal{G}$ is empty. By \autoref{lemma:synchc-corr},
  $\mathcal{G}$ simulates \cc{P}, and with \autoref{lemma:pf-corr},
  this proves \cc{P} is safe.

  \sys returns $\false$ only if a derivation and model $(D,m,i)$ of
  $\mathcal{G}$ is found for which $\isfeasible(\cc{P},p) =
  \true$, where $p$ is the control graph induced by $(D,m,i)$. (And
  $p$ is path of \cc{P} because $\mathcal{G}$ is a path grammar of
  \cc{P}.)  By \autoref{lemma:is-feasible-corr}, $p$ is feasible run
  of \cc{P}, witnessing that \cc{P} is not safe.
\end{proof}

\end{document}